\documentclass[a4paper,11pt,reqno]{amsart}

\usepackage[margin=1in,includehead,includefoot]{geometry}
\usepackage[utf8]{inputenc}
\usepackage[T1]{fontenc}
\usepackage{lmodern,microtype}
\usepackage{hyperref}
\usepackage{url}
\usepackage{booktabs}
\usepackage{mathtools,amsfonts,amsmath,amssymb,amsthm}
\usepackage[foot]{amsaddr}
\usepackage[b]{esvect} % for arrow symbols
\usepackage{graphicx, color}
\usepackage{wrapfig}
\usepackage[margin=5mm]{caption}
\usepackage{verbatim}
\usepackage{subcaption}
\usepackage{enumitem}
\usepackage{tikz}
\usepackage[disable]{todonotes}
\usepackage{mdframed}
\usepackage{xpatch}
\usepackage{mycommands}

\graphicspath{{./graphics/}}

\definecolor{p2col4}{rgb}{0.9, 0.0, 0.0}
\definecolor{p2col3}{rgb}{1, 0.5, 0}
\definecolor{p2col2}{rgb}{0.0, 0.6, 0.1}
\definecolor{p2col1}{rgb}{0, 0, 0.7}
\newcommand{\red}{\color{p2col4}}
\newcommand{\orange}{\color{p2col3}}
\newcommand{\green}{\color{p2col2}}
\newcommand{\blue}{\color{p2col1}}

\usetikzlibrary{arrows,shapes,graphs}
\tikzstyle{vertex}=[circle,draw=black,fill=white,minimum size=13pt,inner sep=0pt]
\tikzstyle{edge} = [draw,thick,-]
\tikzstyle{oriented edge} = [draw,line width=1pt,->]

\newtheorem{theorem}{Theorem}
\newtheorem{lemma}[theorem]{Lemma}

\newtheorem{corollary}[theorem]{Corollary}

\theoremstyle{remark}


\hypersetup{
  pdftitle={Combinatorial generation via permutation languages. IV. Elimination trees},
  pdfauthor={Jean Cardinal, Arturo Merino, Torsten M\"utze}
}

\title{Combinatorial generation via permutation languages. \\ IV. Elimination trees}

\author{Jean Cardinal}
\address[Jean Cardinal]{Computer Science Department, Universit\'e Libre de Bruxelles (ULB), Belgium}
\email{jcardin@ulb.ac.be}

\author{Arturo Merino}
\address[Arturo Merino]{Department of Mathematics, TU Berlin, Germany}
\email{merino@math.tu-berlin.de}

\author{Torsten M\"utze}
\address[Torsten M\"utze]{Department of Computer Science, University of Warwick, United Kingdom}
\email{torsten.mutze@warwick.ac.uk}

\thanks{An extended abstract of this paper appeared in the Proceedings of the 22nd SIAM-ACM Symposium on Discrete Algorithms (SODA) 2022~\cite{DBLP:conf/soda/CardinalMM22}.}

\thanks{Arturo Merino was supported by ANID Becas Chile 2019-72200522.
Torsten M\"utze is also affiliated with the Faculty of Mathematics and Physics, Charles University Prague, Czech Republic, and he was supported by Czech Science Foundation grant GA~19-08554S. Arturo Merino and Torsten M\"utze were also supported by German Science Foundation grant~413902284.}
\begin{document}
\maketitle
\thispagestyle{empty}

\begin{abstract}
An elimination tree for a connected graph~$G$ is a rooted tree on the vertices of~$G$ obtained by choosing a root~$x$ and recursing on the connected components of~$G-x$ to produce the subtrees of~$x$.
Elimination trees appear in many guises in computer science and discrete mathematics, and they encode many interesting combinatorial objects, such as bitstrings, permutations and binary trees.
We apply the recent Hartung-Hoang-M\"utze-Williams combinatorial generation framework to elimination trees, and prove that all elimination trees for a chordal graph~$G$ can be generated by tree rotations using a simple greedy algorithm.
This yields a short proof for the existence of Hamilton paths on graph associahedra of chordal graphs.
Graph associahedra are a general class of high-dimensional polytopes introduced by Carr, Devadoss, and Postnikov, whose vertices correspond to elimination trees and whose edges correspond to tree rotations.
As special cases of our results, we recover several classical Gray codes for bitstrings, permutations and binary trees, and we obtain a new Gray code for partial permutations.
Our algorithm for generating all elimination trees for a chordal graph~$G$ can be implemented in time~$\cO(\sigma)$ on average per generated elimination tree, where $\sigma=\sigma(G)$ denotes the maximum number of edges of an induced star in~$G$.
If $G$ is a tree, we improve this to a loopless algorithm running in time~$\cO(1)$ per generated elimination tree.
We also prove that our algorithm produces a Hamilton cycle on the graph associahedron of~$G$, rather than just Hamilton path, if the graph~$G$ is chordal and 2-connected.
Moreover, our algorithm characterizes chordality, i.e., it computes a Hamilton path on the graph associahedron of~$G$ if and only if $G$ is chordal.
\end{abstract}

\section{Introduction}

Many recent developments in theoretical computer science and combinatorics are closely intertwined.
Specifically, many combinatorial questions are motivated by applications to algorithm design, data structures, or network analysis.
Conversely, most fundamental computational problems involve finite classes of combinatorial objects, such as relations, graphs, or words, and their analysis is a major drive for the development of combinatorial insights.
There are four recurring fundamental algorithmic tasks that we wish to perform with such objects, namely to \emph{count} them or to \emph{sample} one of them at random, to \emph{search} for an object that maximizes some objective function (combinatorial optimization), or to produce an exhaustive \emph{list} of all the objects.
A great deal of literature is devoted to all of these problems, and in this paper we focus on the last and most fine-grained of these tasks, namely \emph{combinatorial generation}.

\subsection{Combinatorial generation}

The complexity of a combinatorial generation algorithm is typically measured as the time it takes to produce the next object in the list from the previous one.
Clearly, the best we can hope for is that each object is produced in constant time.
For this to be possible, any two consecutive objects should not differ much, so that the algorithm can perform the required modification in constant time.
Such a listing of objects subject to some closeness condition is referred to as a \emph{Gray code}~\cite{MR1491049,MR3523863}. 
For some applications a cyclic Gray code is desirable, i.e., the last object in the list and the first one also satisfy the closeness condition.

For example, the classical \emph{binary reflected Gray code}~\cite{gray_1953} is a listing of all bitstrings of length~$n$ such that each string differs from the previous one in a single bit, and this listing is cyclic.

\begin{wraptable}{r}{0.4\textwidth}
\centering
\renewcommand{\arraystretch}{0.9}
\setlength\tabcolsep{6pt}
\small
\vspace{-1mm}
\begin{tabular}{|llll|}
\hline
$n=1$ & $n=2$ & $n=3$ & $n=4$ \\
\hline
{\blue 1} & {\blue 1}{\green 2} & {\blue 1}{\green 2}{\orange 3} & {\blue 1}{\green 2}{\orange 3}{\red 4} \\
          &                     &                                & {\blue 1}{\green 2}{\red 4}{\orange 3} \\
          &                     &                                & {\blue 1}{\red 4}{\green 2}{\orange 3} \\
          &                     &                                & {\red 4}{\blue 1}{\green 2}{\orange 3} \\
          &                     & {\blue 1}{\orange 3}{\green 2} & {\red 4}{\blue 1}{\orange 3}{\green 2} \\
          &                     &                                & {\blue 1}{\red 4}{\orange 3}{\green 2} \\
          &                     &                                & {\blue 1}{\orange 3}{\red 4}{\green 2} \\
          &                     &                                & {\blue 1}{\orange 3}{\green 2}{\red 4} \\
          &                     & {\orange 3}{\blue 1}{\green 2} & {\orange 3}{\blue 1}{\green 2}{\red 4} \\
          &                     &                                & {\orange 3}{\blue 1}{\red 4}{\green 2} \\
          &                     &                                & {\orange 3}{\red 4}{\blue 1}{\green 2} \\
          &                     &                                & {\red 4}{\orange 3}{\blue 1}{\green 2} \\
          & {\green 2}{\blue 1} & {\orange 3}{\green 2}{\blue 1} & {\red 4}{\orange 3}{\green 2}{\blue 1} \\
          &                     &                                & {\orange 3}{\red 4}{\green 2}{\blue 1} \\
          &                     &                                & {\orange 3}{\green 2}{\red 4}{\blue 1} \\
          &                     &                                & {\orange 3}{\green 2}{\blue 1}{\red 4} \\
          &                     & {\green 2}{\orange 3}{\blue 1} & {\green 2}{\orange 3}{\blue 1}{\red 4} \\
          &                     &                                & {\green 2}{\orange 3}{\red 4}{\blue 1} \\
          &                     &                                & {\green 2}{\red 4}{\orange 3}{\blue 1} \\
          &                     &                                & {\red 4}{\green 2}{\orange 3}{\blue 1} \\
          &                     & {\green 2}{\blue 1}{\orange 3} & {\red 4}{\green 2}{\blue 1}{\orange 3} \\
          &                     &                                & {\green 2}{\red 4}{\blue 1}{\orange 3} \\
          &                     &                                & {\green 2}{\blue 1}{\red 4}{\orange 3} \\
          &                     &                                & {\green 2}{\blue 1}{\orange 3}{\red 4} \\ \hline
\end{tabular}
\vspace{2mm}
\caption{The Steinhaus-Johnson-Trotter Gray code for permutations.}
\vspace{-6mm}
\label{tab:SJT}
\end{wraptable}
Another example is the problem of listing all permutations of length~$n$ such that every permutation is obtained from the previous one by an adjacent transposition, i.e., by swapping two neighboring entries of the permutation.
This is achieved by the well-known \emph{Steinhaus-Johnson-Trotter algorithm}~\cite{DBLP:journals/cacm/Trotter62,MR0159764,MR0157881}, which guarantees a cyclic listing; see Table~\ref{tab:SJT}.
A third classical example is the Gray code by Lucas, Roelants van Baronaigien, and Ruskey~\cite{MR1239499} which generates all $n$-vertex binary trees by rotations, albeit non-cyclically.
Binary trees are in bijection with many other Catalan objects such as triangulations of a convex polygon, well-formed parenthesis expressions, Dyck paths, etc.~\cite{MR3467982}.
In triangulations of a convex polygon, the rotation operation maps to another simple operation, known as a flip, which removes the diagonal of a convex quadrilateral formed by two triangles and replaces it by the other diagonal.

Combinatorial generation algorithms have been devised for many other classes of objects~\cite{MR0366085,MR0443423,MR1788836}, including objects derived from graphs and orders, such as spanning trees of a graph~\cite{smith_1997}, maximal cliques or independent sets of a graph~\cite{MR476582}, perfect matchings of bipartite graphs~\cite{MR1349886}, perfect elimination orderings of chordal graphs~\cite{MR2022580}, and linear extensions~\cite{DBLP:journals/cj/VarolR81,MR1267216} or ideals of partial orders~\cite{MR875784}.

A standard reference on combinatorial generation is Volume~4A of Knuth's series `The Art of Computer Programming'~\cite{MR3444818}.
Generic methods have been proposed, such as the reverse-search technique of Avis and Fukuda~\cite{MR1380066}, the ECO framework of Barcucci, Del Lungo, Pergola, and Pinzani~\cite{MR1717162}, the antimatroid formulation of Pruesse and Ruskey~\cite{MR1267190}, Li and Sawada's reflectable languages~\cite{MR2483809}, the bubble language framework of Ruskey, Sawada, and Williams~\cite{MR2844089}, and Williams' greedy algorithm~\cite{DBLP:conf/wads/Williams13}.

\subsection{Flip graphs and polytopes}

Given any class of combinatorial objects and a `local change' operation between them, the corresponding \emph{flip graph} has as vertices the combinatorial objects, and its edges connect pairs of objects that differ by the prescribed change operation.
Partial orders and lattices are often lurking, such as the Boolean lattice for bitstrings, the weak Bruhat order on permutations, and the Tamari lattice for Catalan families.
Moreover, flip graphs can often be realized geometrically as the 1-skeletons of polytopes, and combinatorial generation for such classes of objects therefore amounts to computing Hamilton paths or cycles on this polytope.
The polytopes associated with the three aforementioned examples of bitstrings, permutations, and binary trees are the hypercube, the permutahedron, and the associahedron, respectively; the latter two are shown in Figure~\ref{fig:asso-perm}.
The associahedron, in particular, has a rich history and literature, connecting computer science, combinatorics, algebra, and topology~\cite{MR928904,MR2108555,MR2321739,MR3197650}.

\begin{figure}
\centering
\includegraphics[page=1]{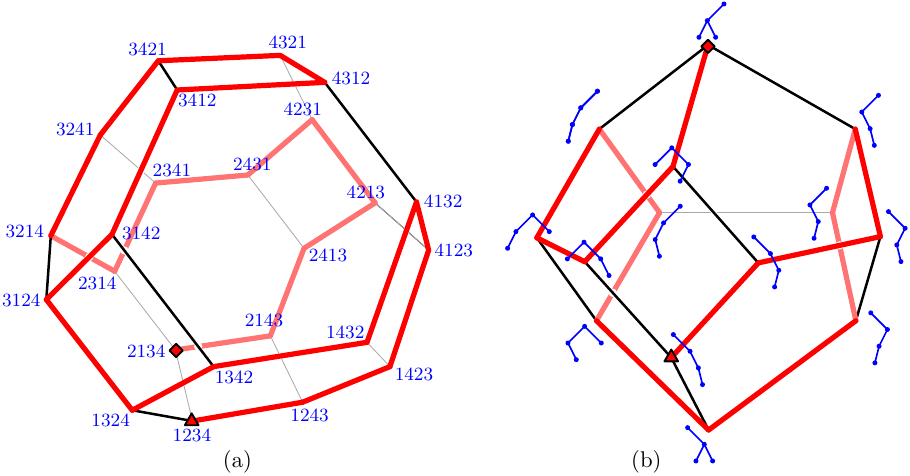}
\caption{The three-dimensional permutahedron (left) and associahedron (right), with the Steinhaus-Johnson-Trotter Hamilton path and the Lucas-Roelants van Baronaigien-Ruskey Hamilton path, respectively (bold edges).
The starting and end vertices are marked by a triangle or a diamond, respectively.
}
\label{fig:asso-perm}
\end{figure}

\subsection{Elimination trees}
\label{sec:elim}

\begin{wrapfigure}{r}{0.5\textwidth}
\centering
\includegraphics[page=1]{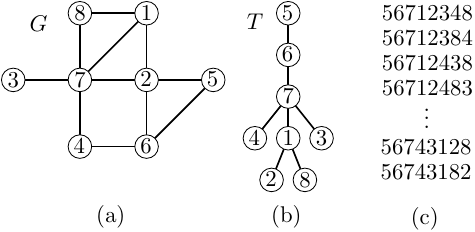}
\caption{(a) A connected graph~$G$; (b) an elimination tree~$T$ for~$G$; (c) elimination orderings yielding the same tree~$T$.}
\label{fig:elim}
\end{wrapfigure}
In this work, we focus on the generation of elimination trees, which are trees on $n$ vertices that are obtained from a fixed graph~$G$ on $n$ vertices, and which capture all ways of removing the vertices of~$G$ one after the other.
For any graph~$G$ and any set of vertices~$X$ we write $G-X$ for the graph obtained by removing every vertex of~$X$ from~$G$.
For a singleton~$X=\{x\}$ we simply write $G-\{x\}=:G-x$.

Given a connected graph~$G=(V,E)$, an \emph{elimination tree} for~$G$ is a rooted tree with vertex set~$V$, composed of a root~$x\in V$ that has as children elimination trees for each connected component of~$G-x$.
This definition is illustrated in Figure~\ref{fig:elim}.
An \emph{elimination forest} for a graph~$G$ is a set of elimination trees, one for each connected component of~$G$.
We write $\cE(G)$ for the set of all elimination forests for~$G$.

We emphasize that an elimination tree is unordered, i.e., there is no ordering associated with the children of a vertex in the tree.
Similarly, there is no ordering among the elimination trees in an elimination forest.
It is useful to think of an elimination tree for a graph~$G$ as the outcome of the process of removing vertices in some \emph{elimination ordering}, which is a permutation that specifies the order of removed vertices; see Figure~\ref{fig:elim}~(c):
We first remove the root~$x$ from~$G$, then proceed to remove the next vertex in the ordering from the connected component of $G-x$ it belongs to.
In general, one elimination tree corresponds to several distinct elimination orderings.
Specifically, these are all the linear extensions of the partial order whose cover graph is the elimination tree turned upside down.

\subsection{Applications and related notions}
\label{sec:appl}

Elimination trees are also found under the guise of vertex rankings and centered colorings, and elimination forests are also known as $G$-forests~\cite{MR4176852}, spines~\cite{MR3383157}, and when defined in the more general context of building sets, as $\cB$-forests~\cite{MR2487491}.
They have been studied extensively in various contexts, including data structures, combinatorial optimization, graph theory, and polyhedral combinatorics.

For example, Liu and coauthors~\cite{MR937485,MR1062491,MR1032223,MR2137477,MR2369300} used elimination trees in efficient parallel algorithms for matrix factorization. 
Elimination trees are also met in the context of VLSI design~\cite{DBLP:conf/focs/Leiserson80,MR1187395}, and for parallel scheduling in modular products manufacturing~\cite{iyer_ratliff_vijayan_1988,MR1090666,nevins_whitney_1989}. 
In the context of scheduling, one is typically interested in finding an elimination tree of minimum height, which determines the number of parallel steps in the schedule.
This problem, known to be NP-hard in general, has drawn a lot of attention in the last thirty years~\cite{scheffler_1993,MR1430905,MR1288578,MR1723686,MR1317666,MR1612885}. 
Computing optimal elimination trees for trees~$G$ is possible in linear time~\cite{MR958825,MR1031599}.

A central notion in graph theory is the \emph{tree-depth} of a graph, which is yet another name for the minimum height of an elimination tree~\cite{MR2920058,MR3238682,MR3361785,MR3775805}.
In particular, tree-depth and elimination trees can be defined via the following other well-known objects.
A \emph{ranking} of the vertices of a graph~$G$ is a labeling of its vertices with integers from~$\{1,2,\ldots,k\}$ such that any path between two vertices with the same label contains a vertex with a larger label.
A \emph{centered coloring} of a graph~$G$ is a vertex coloring such that for any connected subgraph~$H$, some color appears exactly once in~$H$.
It is not difficult to show that the minimum~$k$ for which there exists a vertex ranking of~$G$ is equal to the minimum number of colors in a centered coloring of~$G$, which is in turn equal to the tree-depth of~$G$.
For a connected graph~$G$, the elimination tree corresponding to a vertex ranking or a centered coloring can be constructed by iteratively picking respectively the largest label or the unique color as the root~$x$ of the tree, and recursing on the connected components of~$G-x$.

Elimination trees also occur naturally in the problem of searching in a tree or a graph~\cite{MR1699005,DBLP:conf/focs/OnakP06,MR2487681,MR3536593}, with applications to fault detection and database integrity checking.
Recently, an online search model on trees was defined based on elimination trees~\cite{MR4141297}, which generalizes~\cite{MR4415112} the classical splay tree data structure of Sleator and Tarjan.

\subsection{Encoded combinatorial objects}
\label{sec:objects}

\begin{figure}[t!]
\centering
\includegraphics[page=3]{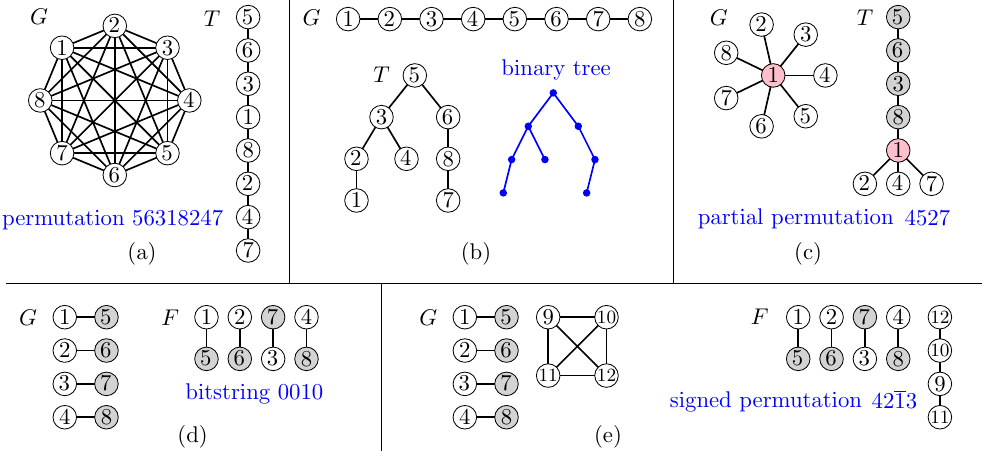}
\caption{Combinatorial objects encoded by elimination trees for suitable graphs~$G$.}
\label{fig:objects}
\end{figure}

In the context of combinatorial generation, elimination trees are interesting, as they encode several familiar combinatorial objects:
\begin{itemize}[leftmargin=5mm, noitemsep, topsep=3pt plus 3pt]
\item When $G$ is the complete graph on~$[n]$, then its elimination trees are paths, which can be interpreted as \emph{permutations} of~$[n]$: Read off the vertex labels from the root to the leaf in the elimination tree; see Figure~\ref{fig:objects}~(a).
\item When $G$ is the path with vertices labeled $1,\ldots,n$ between the end vertices, then its elimination trees are all $n$-vertex \emph{binary trees}: The distinction between left and right child in the binary tree is induced by the smaller and larger vertex labels; see Figure~\ref{fig:objects}~(b).
\item When $G$ is a star with~1 as the center and with leaves $2,\ldots,n$, then its elimination trees are brooms: a path composed of elements from a subset of~$[n]\setminus\{1\}$, followed by a subtree of height one rooted in~1.
By reading off the labels from the handle of the broom starting at the root and ending at the parent of~1, and subtracting 1 from those labels, we obtain a linearly ordered subset of~$[n-1]$, which is known as a \emph{partial permutation}; see Figure~\ref{fig:objects}~(c).
We see that elimination trees for stars are in one-to-one correspondence with partial permutations.
\item The graph~$G$ may also be disconnected.
In particular, if $G$ is a disjoint union of $n$ edges $\{i,n+i\}$ for $i=1,\ldots,n$, then its elimination forests consist of $n$ disjoint one-edge trees, which are either rooted in~$i$ or~$n+i$ for all $i=1,\ldots,n$.
We can thus interpret the elimination forest as a \emph{bitstring} of length~$n$, where the $i$th bit is 0 if $i$ is root, and the $i$th bit is 1 if $n+i$ is root; see Figure~\ref{fig:objects}~(d).
\item Combining the aforementioned encodings for permutations and bitstrings, we can take $G$ as a disjoint union of $n$ edges $\{i,n+i\}$ for $i=1,\ldots,n$ and a complete graph on the $n$ vertices~$\{2n+1,\ldots,3n\}$.
The elimination forests for~$G$ can be interpreted as \emph{signed permutations} of~$[n]$: Read off the vertex labels of the path on~$\{2n+1,\ldots,3n\}$ from the root to the leaf in the corresponding elimination tree, subtracting $2n$ from those labels, and take the resulting entry~$i$ of the permutation with positive sign if $i$ is root and with negative sign if $n+i$ is root; see Figure~\ref{fig:objects}~(e).
\end{itemize}
The task of generating all elimination trees for a graph considered in this paper is thus a generalization of generating each of the aforementioned concrete classes of combinatorial objects.

\subsection{Rotations and graph associahedra}
\label{sec:rot}

Elimination trees can be locally modified by rotation operations, which generalize the binary tree rotations used in standard online binary search tree algorithms~\cite{MR0156719,MR539826,MR796206}.
In fact, rotations are one of the elementary, unit-cost operations in the online search model studied in~\cite{MR4141297,MR4415112}.

Formally, rotations in elimination trees are defined as follows; see Figure~\ref{fig:rotation}.
Let~$T$ be an elimination tree for a connected graph~$G$ and let~$j$ be a vertex from~$G$, distinct from the root of~$T$.
Let~$i$ be the parent of~$j$ in~$T$, and let $H$ be the subgraph of~$G$ induced by the vertices in the subtree rooted at~$i$.
Then the \emph{rotation of the edge~$\{i,j\}$} transforms~$T$ into another elimination tree~$T'$ for~$G$ in which:
\begin{itemize}[leftmargin=5mm, noitemsep, topsep=5pt]
\item $j$ becomes the parent of~$i$, and the child of the parent of~$i$ in~$T$ (or the root if~$i$ is the root of~$T$),
\item the subtrees of~$i$ in~$T$ remain subtrees of~$i$,
\item a subtree~$S$ of~$j$ in~$T$ remains a subtree of~$j$, unless the vertices of~$S$ belong to the same connected component of~$H-j$ as~$i$, in which case~$S$ becomes a subtree of~$i$.
\end{itemize}

\begin{wrapfigure}{r}{0.5\textwidth}
\centering
\begin{tabular}{cc}
\begin{tikzpicture}[scale=.8,auto,swap]
\foreach \pos/\name in {{(1,2)/a}, {(2,1)/b}, {(2,3.5)/c}, {(3.5,2)/d}, {(3.5,4.5)/e}, {(4.5,3.5)/f}} \node[vertex] (\name) at \pos {$\name$};
\node[vertex, fill=red!20] (p) at (2,2) {$j$};
\node[vertex, fill=red!20] (v) at (3.5,3.5) {$i$};
\foreach \source/\dest in {a/p,b/p,p/c,p/d,c/v,v/d,v/e,v/f} \path[edge] (\source) -- (\dest);
\node at (1,4) {$G$};
\end{tikzpicture}
&
\begin{tikzpicture}[scale=.8,auto,swap]
\begin{scope}[yshift=-2cm]
\foreach \pos/\name in {{(0.5,1)/a}, {(1.5,1)/b}, {(2.5,1)/c}, {(3.5,1)/d}, {(3.5,2)/e}, {(4.5,2)/f}} \node[vertex] (\name) at \pos {$\name$};
\node[vertex, fill=red!20] (p) at (2,2) {$j$};
\node[vertex, fill=red!20] (v) at (3.5,3) {$i$};
\foreach \source/\dest in {p/a,p/b,p/v,p/c,p/d,v/e,v/f} \path[edge] (\source) -- (\dest);
\node at (5,2.9) {$T$};
\end{scope}
\begin{scope}[xshift=-0.5cm,yshift=-6cm]
\draw[oriented edge] (3.5,4.3) -- (3.5,3.2);
\foreach \pos/\name in {{(1.5,2)/a}, {(2.5,2)/b}, {(2.5,1)/c}, {(3.5,1)/d}, {(4.5,1)/e}, {(5.5,1)/f}} \node[vertex] (\name) at \pos {$\name$};
\node[vertex, fill=red!20] (p) at (2.5,3) {$j$};
\node[vertex, fill=red!20] (v) at (4,2) {$i$};
\foreach \source/\dest in {p/a,p/b,p/v,v/c,v/d,v/e,v/f} \path[edge] (\source) -- (\dest);
\node at (5.5,2.9) {$T'$};
\end{scope}
\end{tikzpicture} \\
(a) & (b)
\end{tabular}
\caption{Elimination tree rotation. Part~(a) shows two vertices~$i$ and~$j$ in a graph~$G$; (b)~shows the corresponding tree rotation.}
\vspace{-3mm}
\label{fig:rotation}
\end{wrapfigure}
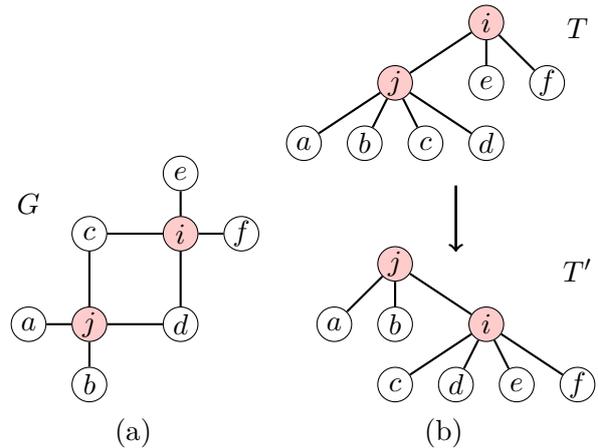
A rotation in an elimination forest for a disconnected graph is a rotation in one of its elimination trees.
A rotation can be interpreted as the change in an elimination tree for~$G$ resulting from swapping~$i$ and~$j$ in the elimination ordering of the vertices.

Under the encodings discussed in Section~\ref{sec:objects}, elimination tree rotations correspond to natural `local change' operations on the corresponding combinatorial objects.
Specifically, one can check that they translate to adjacent transpositions in permutations, classical rotations in binary trees, adjacent transpositions or deletions or insertions of a trailing element in partial permutations, flipping a single bit in bitstrings, or adjacent transpositions or sign changes in signed permutations, respectively.

It is well known that for any graph~$G$, the flip graph of elimination forests for~$G$ under tree rotations is the skeleton of a polytope, referred to as the \emph{graph associahedron}~$\cA(G)$~\cite{MR2239078,MR2479448,MR2487491}.
Graph associahedra are special cases of \emph{generalized permutahedra} that have applications in algebra and physics~\cite{MR2520477,aguiar_ardila_2023}.
For $G$ being a complete graph, a cycle, a path, a star, or a disjoint union of edges, $\cA(G)$ is the permutahedron, the cyclohedron, the standard associahedron, the stellohedron, or the hypercube, respectively.
Figure~\ref{fig:asso4} shows the graph associahedra of all 4-vertex graphs.

\begin{figure}
\centering
\includegraphics[width=13cm]{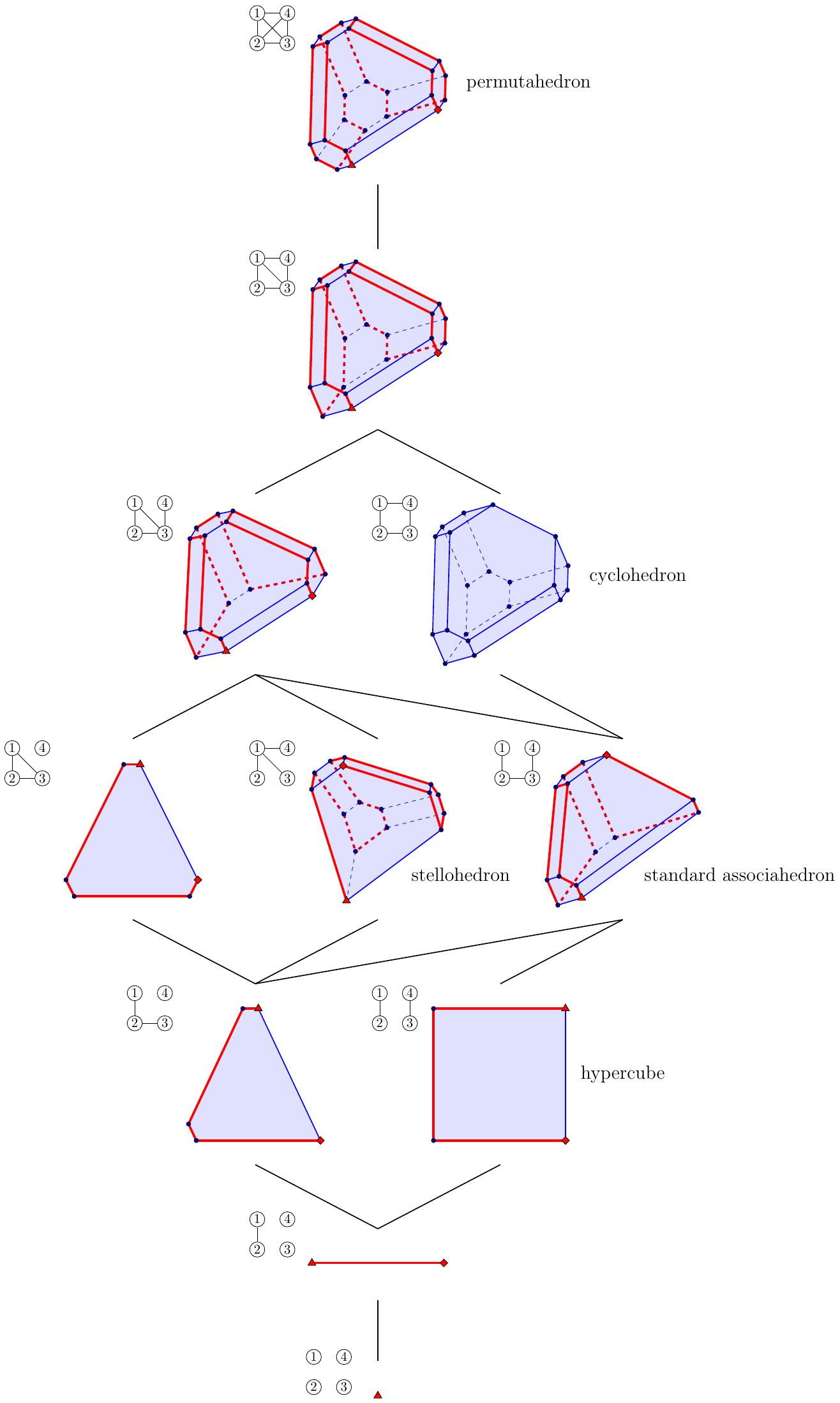}
\caption{Graph associahedra~$\cA(G)$ for all graphs~$G$ on $n=4$ vertices, ordered by subgraph inclusion.
The Hamilton paths computed by our algorithm for all chordal graphs~$G$ are highlighted, with the starting and end vertex marked by a triangle and diamond, respectively.
The only non-chordal graph is the 4-cycle.
}
\label{fig:asso4}
\end{figure}

We consider the problem of generating all elimination forests for a graph~$G$ by rotations, or equivalently, of computing Hamilton paths and cycles on the graph associahedron~$\cA(G)$.
In previous work, Manneville and Pilaud~\cite{MR3383157} showed that for any graph~$G$ with at least two edges, $\cA(G)$ has a Hamilton cycle.
Their construction is an inductive gluing argument on~$\cA(G)$, which does not translate into an efficient algorithm for computing such a cycle.
Note that the number of vertices of~$\cA(G)$ is in general exponential in the number~$n$ of vertices of the underlying graph~$G$ (for example, the permutahedron has $n!$ vertices), which makes global manipulations on~$\cA(G)$ prohibitive for combinatorial generation, where we aim for an algorithm that visits each vertex of~$\cA(G)$ in time polynomial in~$n$, ideally even constant.

To obtain such an efficient algorithm, we apply the combinatorial generation framework recently proposed by Hartung, Hoang, M\"utze, and Williams~\cite{MR4391718}.
In this framework, the objects to be generated are encoded by permutations, and those permutations are generated by a simple greedy algorithm.
Our encoding considers for each elimination tree of an $n$-vertex graph the set of all elimination orderings (=permutations of~$[n]$) corresponding to this tree (recall Figure~\ref{fig:elim}~(b)+(c)), and fixes precisely one representative permutation from this set.
These representatives are chosen so that their union, which is a subset of all permutations of~$[n]$, forms a so-called \emph{zigzag language}, a term defined in~\cite{MR4391718} via a closure property.
The algorithm proposed in that paper to generate zigzag languages and the combinatorial objects they encode can be implemented efficiently for many classes of objects, and it subsumes several previously studied Gray codes.
In a series of recent papers, this framework was applied to a plethora of combinatorial objects such as pattern-avoiding permutations~\cite{MR4391718}, lattice quotients of the weak order on permutations~\cite{MR4344032}, and rectangulations~\cite{MR4598046}.
In this work, we extend the reach of this framework and make it applicable to the efficient generation of structures on graphs, specifically of elimination forests, which is a step forward in exploring the generality of this approach.
This is achieved by combining algorithmic, combinatorial, and polytopal insights and methods.

\subsection{Our results}

In the following we summarize the main results of this work and sketch the main ideas for proving them.

\subsubsection{A simple algorithm for generating elimination forests for chordal graphs}
\label{sec:results-algo}

For our algorithm it is convenient to encode the rotation of edges~$\{i,j\}$, $i<j$, by the larger end vertex~$j$ of the edge, and by the direction in which~$i$ is reached from~$j$, namely upwards if $i$ is the parent of~$j$ and downwards if $i$ is a child of~$j$.
This is the direction in which the vertex~$j$ moves as a result of the rotation.
We refer to these operations as \emph{up- and down-rotations of~$j$}, respectively, and we use the shorthand notations~$j{\diru}$ and~$j{\dird}$.
Observe that a down-rotation $j{\dird}$ is only well-defined if $j$ has a unique child that is smaller than~$j$, otherwise there are several choices for children~$i<j$ of~$j$ and consequently several possible edges to rotate.

We propose to generate the set~$\cE(G)$ of all elimination forests for a graph~$G=([n],E)$, $[n]:=\{1,2,\ldots,n\}$, using the following simple greedy algorithm.

\begin{algo}{Algorithm~R}{Greedy rotations}
\label{algo:JE}
This algorithm attempts to greedily generate the set~$\cE(G)$ of elimination forests for a graph $G=([n],E)$ using rotations starting from an initial elimination forest $F_0\in\cE(G)$.
\begin{enumerate}[label={\bfseries R\arabic*.}, leftmargin=8mm, noitemsep, topsep=3pt plus 3pt]
\item{} [Initialize] Visit the initial elimination forest~$F_0$.
\item{} [Rotate] Generate an unvisited elimination forest from~$\cE(G)$ by performing an up- or down-rotation of the largest possible vertex in the most recently visited elimination forest.
If no such rotation exists, or the rotation edge is ambiguous, then terminate.
Otherwise, visit this elimination forest and repeat~R2.
\end{enumerate}
\end{algo}

In other words, we consider the vertices $n,n-1,\ldots,2$ of the current elimination forest in decreasing order, and for each of them we check whether it allows an up- or down-rotation that creates a previously unvisited elimination forest, and we perform the first such rotation we find, unless the same vertex allows several possible rotations, in which case we terminate.
We also terminate if no rotation creates an unvisited elimination forest.

\begin{figure}
\centering
\includegraphics[page=6]{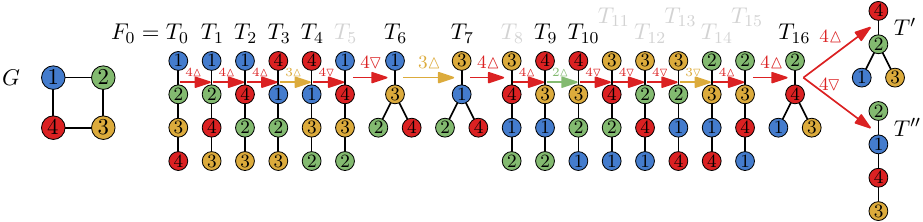}
\caption{The output $T_0,\ldots,T_{16}$ of Algorithm~R for the 4-cycle.}
\label{fig:C4}
\end{figure}

For example, consider all elimination trees for the 4-cycle with vertices labeled $1,2,3,4$ cyclically; see Figure~\ref{fig:C4}.
When initialized with the elimination tree~$F_0=T_0$ shown in the figure, the algorithm visits the 17 elimination trees~$T_0,\ldots,T_{16}$.
The tree~$T_0$ admits an up-rotation of~4, yielding~$T_1$.
The tree~$T_1$ admits an up- and down-rotation of~4, but the latter would yield~$T_0$, which was already visited, so we perform $4{\diru}$, yielding~$T_2$.
One more up-rotation of~4 gives~$T_3$, which does not admit any rotations of~4 to unvisited elimination trees.
Consequently, we consider the vertex~3, which does admit an up-rotation, yielding~$T_4$.
The next interesting step is~$T_6$, which does not admit rotations of~4 to unvisited elimination trees.
However, $T_6$ admits an up- and down-rotation of~3, but the latter would lead to~$T_0$ again, so we perform $3{\diru}$ to reach $T_7$.
From $T_9$ to $T_{10}$ we up-rotate 2, as neither 4 nor 3 admit rotations to unvisited elimination trees.
The algorithm eventually terminates with $T_{16}$, which admits both an up- and down-rotation of~4 to two previously unvisited elimination trees~$T'$ and~$T''$.
Because of this ambiguity, the algorithm terminates without exhaustively generating all elimination trees for~$G$.

Figure~\ref{fig:algo} shows the output of Algorithm~R for four other graphs~$G$, and in all those cases the algorithm terminates because from the last elimination forest in those lists, no rotation leads to a previously unvisited elimination forests.
Moreover, those four lists are all exhaustive, i.e., the algorithm succeeds in generating \emph{all} elimination forests for those graphs.

Our main result is that Algorithm~R succeeds to generate~$\cE(G)$ exhaustively for \emph{chordal} graphs~$G$, i.e., graphs in which every induced cycle has length three.
Chordal graphs include many interesting subclasses, such as paths, stars, trees, $k$-trees, complete graphs, interval graphs, and split graphs (in particular, all the graph classes mentioned in Section~\ref{sec:objects}).
A classical characterization of chordal graphs is that they have \emph{perfect elimination ordering}, i.e., a linear ordering of their vertices such that every vertex~$x$ induces a clique together with its neighbors in the graph that come before~$x$ in the ordering.
In what follows, we consider a chordal graph~$G=([n],E)$, where the ordering $1,2,\ldots,n$ is a perfect elimination ordering of~$G$.

\begin{theorem}
\label{thm:jump-elim}
Given any chordal graph~$G=([n],E)$ in perfect elimination order, Algorithm~R visits every elimination forest from~$\cE(G)$ exactly once, when initialized with the elimination forest~$F_0$ that is obtained by removing vertices in increasing order.
\end{theorem}

Theorem~\ref{thm:jump-elim} thus provides a short proof that the graph associahedron~$\cA(G)$ has a Hamilton path for chordal graphs~$G$.
Figure~\ref{fig:asso4} shows the Hamilton paths on the graph associahedra for all chordal 4-vertex graphs computed by our algorithm.

\begin{figure}
\centering
\makebox[0cm]{ % artificial box to center the picture
\includegraphics[page=5]{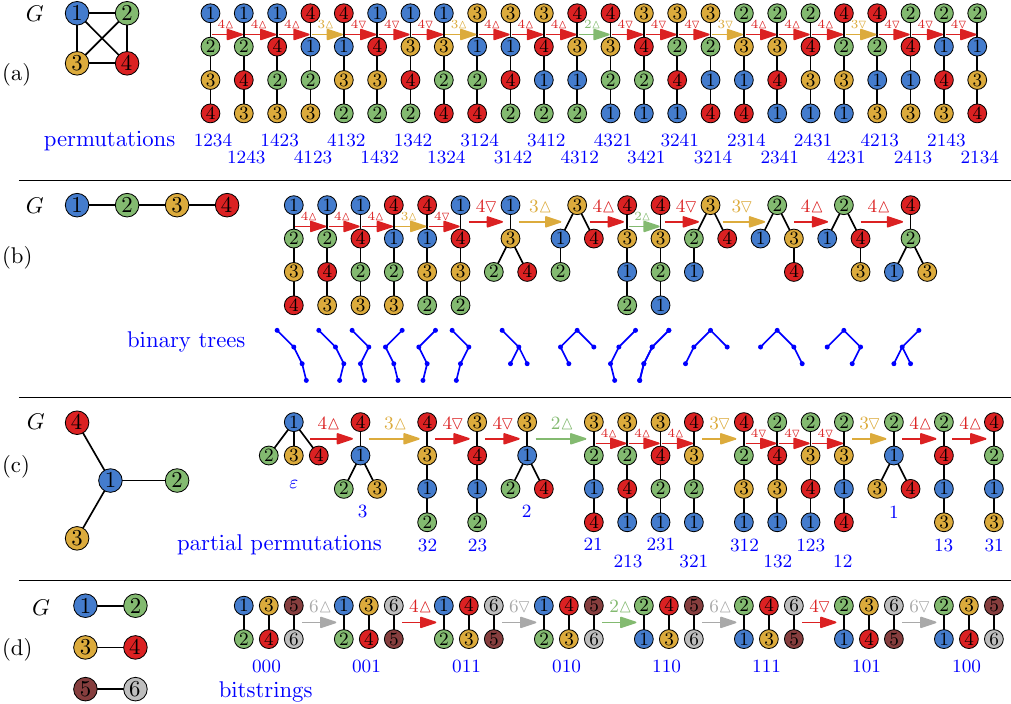}
}
\caption{The output of Algorithm~R for four different chordal graphs~$G$, and the corresponding Gray codes of combinatorial objects.}
\label{fig:algo}
\end{figure}

Algorithm~R generalizes several known Gray codes, including the Steinhaus-Johnson-Trotter algorithm for permutations (when $G$ is a complete graph; see Figure~\ref{fig:algo}~(a) and Table~\ref{tab:SJT}), the binary tree Gray code due to Lucas, Roelants van Baronaigien, and Ruskey (when $G$ is a path; see Figure~\ref{fig:algo}~(b)), and the binary reflected Gray code for bitstrings (when $G$ is a disjoint union of edges; see Figure~\ref{fig:algo}~(d)).
The Gray code for partial permutations (when $G$ is a star; see Figure~\ref{fig:algo}~(c)) via adjacent transpositions or deletions or insertions of a trailing element is new, and it can be implemented in constant time per generated object (see the next section).

Intuitively, the reason why Algorithm~R succeeds for chordal graphs is that in every elimination forest for a chordal graph, every vertex~$j$ has at most one child that is smaller than~$j$.
We show that this property characterizes chordality, i.e., a graph is chordal if and only if all of its elimination forests have this property.
It ensures that to every vertex, at most one up-rotation and at most one down-rotation is applicable, and if both are applicable, then one of the two resulting elimination forests has been visited before by the algorithm, and hence the other one is visited next.
In other words, there will never be ambiguity about two possible down-rotations of a vertex that lead to unvisited elimination forests, or one possible up-rotation and one down-rotation (as in the last step in Figure~\ref{fig:C4}), so the algorithm does not terminate prematurely.
By definition, the algorithm generates a previously unvisited elimination forest in every step, so avoiding premature termination guarantees that~$\cE(G)$ is generated exhaustively.

In fact, we show that Algorithm~R generates a Hamilton path on the graph associahedron~$\cA(G)$ if and only if $G$ is chordal.
As Algorithm~R is oblivious of the notion of a chordal graph, this is an interesting new characterization of graph chordality.

\subsubsection{Efficient implementation of the algorithm}
\label{sec:results-impl}

When implemented naively, Algorithm~R requires storing all previously visited elimination forests, in order to decide upon the next rotation.
We can get rid of this defect and make the algorithm history-free and efficient.
For any graph~$G$, we let $\sigma=\sigma(G)$ denote the number of edges of the largest induced star in~$G$.
It is easy to see that the number of children in an elimination tree for~$G$, maximized over all elimination trees, is precisely~$\sigma(G)$.

\begin{theorem}
\label{thm:algo}
Algorithm~R can be implemented such that for any chordal graph~$G=([n],E)$ in perfect elimination order, the algorithm visits each elimination forest for~$G$ in time~$\cO(\sigma)$ on average, where $\sigma=\sigma(G)$.
For trees~$G$, this can be improved to~$\cO(1)$ worst-case time for visiting each elimination tree for~$G$.
\end{theorem}

Note that for the complete graph~$G$ on~$n$ vertices we have $\sigma(G)=1$, i.e., our algorithm optimally computes the Steinhaus-Johnson-Trotter Gray code of permutations.
Furthermore, for trees~$G$, the time bound~$\cO(1)$ holds in the worst-case in every iteration, i.e., the obtained algorithm is loopless.
Recall that trees~$G$ are of particular interest in view of the special cases mentioned in Section~\ref{sec:objects} and the data structure applications discussed at the end of Section~\ref{sec:appl}.
The memory and initialization time required for these algorithms is $\cO(n^2)$.
The initialization time includes the time for testing chordality and computing a perfect elimination ordering.

We implemented both of these algorithms in C++, and made the code available for download, experimentation and visualization on the Combinatorial Object Server~\cite{cos_elim}.

To achieve the runtime bounds stated in Theorem~\ref{thm:algo}, we maintain an array of direction pointers~$o=(o_1,\ldots,o_n)$, where an entry~$o_j=\diru$ indicates that the vertex~$j$ is rotating up in its elimination tree when it is rotated next by the algorithm, and $o_j=\dird$ indicates that $j$ is rotating down upon the next rotation.
The direction is reversed if after an up-rotation~$j{\diru}$ the vertex~$j$ has become the root of its elimination tree or its parent is larger than~$j$, or if after a down-rotation~$j{\dird}$ the vertex~$j$ has become a leaf or its children are all larger than~$j$.
In addition, we introduce an array that allows us to determine in constant time which vertex~$j$ is rotating in the next step, i.e., which is the `largest possible vertex' in step~R2 of Algorithm~R.

The $\cO(\sigma)$ time bound for chordal graphs~$G$ comes from the maximum number of subtrees that may change their parent as a result of a tree rotation.
For trees~$G$, every pair of vertices is connected by a unique path, and at most one subtree changes parent upon a tree rotation.
This allows us to obtain an improved loopless~$\cO(1)$ time algorithm.

\subsubsection{Hamilton cycles for 2-connected chordal graphs}
\label{sec:results-cycle}

Lastly, we investigate when Algorithm~R produces a cyclic Gray code, i.e., a Hamilton cycle on the graph associahedron~$\cA(G)$, rather than just a Hamilton path.
We aim to understand under which conditions on~$G$ the first and last elimination forest generated by our algorithm differ in a single tree rotation.
In the examples from Figure~\ref{fig:algo}, this is the case for~(a) and~(d), but not for~(b) and~(c).
We derive a number of such conditions, two of which are summarized in the following theorem.
A graph is \emph{2-connected}, if it has at least three vertices and removing any vertex leaves a connected graph.

\begin{theorem}
\label{thm:cycle-summary}
Let $G=([n],E)$ be a chordal graph in perfect elimination order.
If $G$ is 2-connected, then the rotation Gray code for~$\cE(G)$ generated by Algorithm~R is cyclic.
On the other hand, if $G$ is a tree with at least four vertices, then this Gray code is not cyclic.
\end{theorem}

To appreciate the number of cases covered by our results in Theorems~\ref{thm:jump-elim}--\ref{thm:cycle-summary}, we remark that the number of non-isomorphic $n$-vertex graphs is~$2^{n^2/2(1+o(1))}$~\cite{MR0357214}, and the number of chordal graphs and of 2-connected chordal graphs is~$2^{n^2/4(1+o(1))}$~\cite{MR951781} (with different $o(1)$ terms).

\subsection{Outline of this paper}

In Section~\ref{sec:gasso} we summarize the main definitions and results on graph associahedra relevant for our work.
In Section~\ref{sec:zigzag} we describe the permutation language framework from~\cite{MR4391718}.
Section~\ref{sec:unit-stars} applies the framework to elimination trees of two special classes of chordal graphs, namely unit interval graphs and stars, the latter of which yields our new Gray code for partial permutations.
These two special cases deserve discussion in their own right, and they allow us to get familiar with our tools.
In Section~\ref{sec:chordal} we generalize the algorithm to elimination forests for arbitrary chordal graphs, and we provide the proof of Theorem~\ref{thm:jump-elim}.
In Section~\ref{sec:algo} we describe how to implement this algorithm efficiently, thus proving Theorem~\ref{thm:algo}.
In Section~\ref{sec:cycle} we describe conditions for when our algorithm produces a Hamilton cycle on the graph associahedron, rather than just a Hamilton path, thus proving Theorem~\ref{thm:cycle-summary}.
Lastly, in Section~\ref{sec:necessity} we show that our algorithm characterizes chordality.
We conclude in Section~\ref{sec:open} with some interesting open problems.

\section{Graph associahedra}
\label{sec:gasso}

The notion of graph associahedra generalizes that of standard associahedra, and stems from seminal works of Carr and Devadoss~\cite{MR2239078,MR2479448} and Postnikov~\cite{MR2487491}.
These polytopes are notable special cases of generalized permutahedra and hypergraphic polytopes, which play an important role in combinatorial Hopf algebras~\cite{aguiar_ardila_2023}.

\begin{wrapfigure}{r}{0.35\textwidth}
\centering
\includegraphics[page=2]{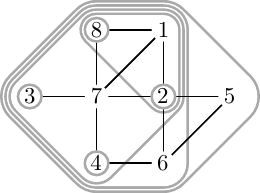}
\caption{The tubing corresponding to the elimination tree shown in Figure~\ref{fig:elim}.}
\vspace{-4mm}
\label{fig:tubing}
\end{wrapfigure}
A \emph{tube} in a graph~$G$ is a nonempty subset of vertices inducing a connected subgraph.
The set of all tubes of~$G$ is the \emph{building set}~$\cB(G)$ of~$G$.
A \emph{tubing} of~$G=([n],E)$ is a set of tubes that includes the vertex sets of each connected component of~$G$,
and such that every pair~$A,B$ of tubes is either
\begin{itemize}[leftmargin=5mm, noitemsep, topsep=3pt plus 3pt]
\item nested, i.e., $A\seq B$ or $A\supseteq B$, or
\item disjoint and non-adjacent, i.e., $A\cup B$ is not a tube of~$G$.
\end{itemize}
See Figures~\ref{fig:tubing} and~\ref{fig:gasso} for illustration.

\begin{figure}
\centering
\includegraphics[page=2,scale=0.9]{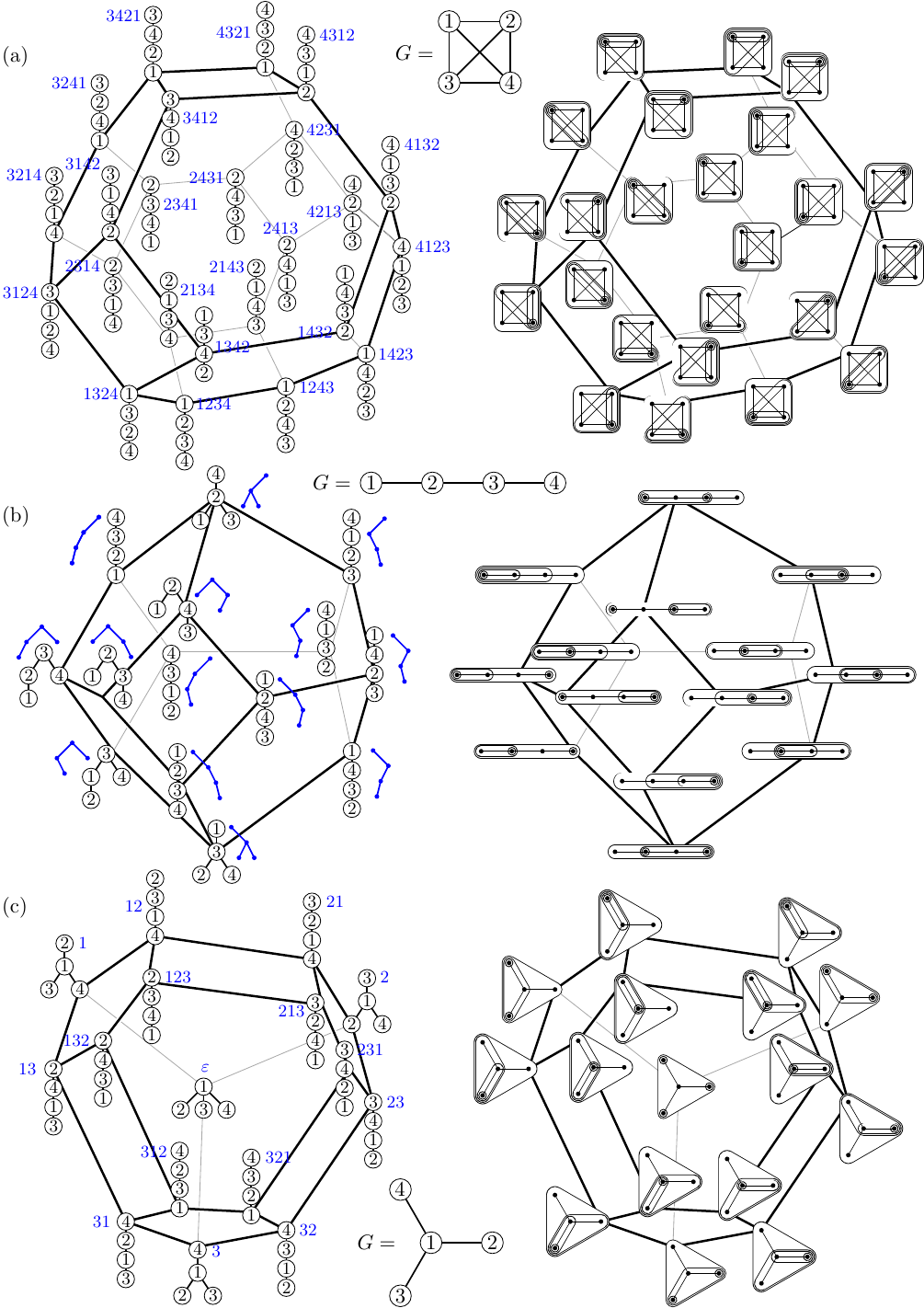}
\caption{Three-dimensional graph associahedra for three graphs on four vertices: (a)~permutahedron; (b)~associahedron; (c)~stellohedron.
Vertices are labeled by elimination trees (left) and the corresponding combinatorial objects (binary trees, permutations and partial permutations, respectively) and by tubings (right).
}
\label{fig:gasso}
\end{figure}

The \emph{graph associahedron of~$G$}, denoted $\cA(G)$, is the polytope whose face lattice is isomorphic to the reverse containment order of the tubings of~$G$.
In particular, vertices of~$\cA(G)$ are in one-to-one correspondence with the inclusionwise maximal tubings of~$G$.
Moreover, note that inclusionwise maximal tubings of~$G$ are in one-to-one correspondence with elimination forests for~$G$.
Indeed, the tubes are sets of vertices contained in the subtrees of the elimination trees for the components of~$G$; see Figure~\ref{fig:elim}~(b) and Figure~\ref{fig:tubing}.

As mentioned before, graph associahedra of paths are standard associahedra; see Figure~\ref{fig:gasso}~(b).
When $G$ is a complete graph, its graph associahedron is the permutahedron; see Figure~\ref{fig:gasso}~(a).
Moreover, graph associahedra of stars are known as \emph{stellohedra}~\cite{MR2520477}; see Figure~\ref{fig:gasso}~(c).
Furthermore, graph associahedra of a disjoint union of edges are \emph{hypercubes}.
Lastly, graph associahedra of cycles are known as \emph{cyclohedra}~\cite{MR1665469,MR1979780,MR2321739}.
The graph associahedra of all 4-vertex graphs are shown in Figure~\ref{fig:asso4}.

We mention an elegant geometric realization of graph associahedra due to Postnikov.

\begin{theorem}[\cite{MR2487491}]
\label{thm:Msum}
The associahedron of a graph~$G=([n],E)$ is realized by the Minkowski sum $\cA(G)=\sum_{S\in\cB(G)} \Delta_S$, where $\Delta_S:=\conv \{e_i \mid i\in S\}$ is the standard simplex associated with a subset $S\seq [n]$, and $e_i$ is the $i$th canonical basis vector of~$\mathbb{R}^n$.
\end{theorem}

For the special case of paths, this realization was given by Loday~\cite{MR2108555}.
In this work, we are mostly interested in the skeleton of the graph associahedron, and not in geometric realizations.
Edges of the (path) associahedron correspond to simple operations on the binary trees that correspond to the vertices, namely rotations in binary trees.
The latter gracefully generalize to elimination forests.

\begin{lemma}[\cite{MR3383157}]
\label{lem:Gasso-edges}
Two elimination forests for a graph~$G$ differ by a single rotation in one of their elimination trees if and only if they correspond to endpoints of an edge in~$\cA(G)$.
\end{lemma}

Our goal is to give a simple and efficient algorithm for generating all elimination forests for~$G$ by rotations.
From the previous lemma, this amounts to computing a Hamilton path on the skeleton of~$\cA(G)$.
In fact, Manneville and Pilaud showed that the skeletons of all graph associahedra have a Hamilton cycle.

\begin{theorem}[\cite{MR3383157}]
\label{thm:cycle}
For any graph~$G$ with at least two edges, the graph associahedron~$\cA(G)$ has a Hamilton cycle.
\end{theorem}

However, as discussed in Section~\ref{sec:rot}, their proof does not yield an efficient method for computing a Hamilton cycle or path on~$\cA(G)$.

Apart from Hamilton cycles, another property of interest is the \emph{diameter} of graph associahedra.
This is the minimum number of rotation operations that are sufficient and sometimes necessary to transform any two elimination forests for~$G$ into each other.
This quantity is known precisely when $G$ is a path~\cite{MR928904,MR3197650}, and bounds are known for various other cases~\cite{MR3649601,MR3874284,MR4505048,MR4473453}.

\section{Zigzag languages of permutations}
\label{sec:zigzag}

In order to apply the combinatorial generation framework proposed by Hartung et al.~\cite{MR4391718} to elimination forests, we need to encode elimination forests by permutations, and ensure that the resulting set of permutations satisfies a certain closure property.
We can then use a simple greedy algorithm that is guaranteed to generate all permutations from the set.
In this section, we first describe the encoding of elimination forests by permutations, and we then state the main result of~\cite{MR4391718}.

\subsection{Encoding elimination forests by permutations}
\label{sec:encoding}

In this section we make the correspondence between elimination trees and permutations mentioned in Section~\ref{sec:elim} precise.

We write $S_n$ for the set of permutations of~$[n]$.
When viewed as a partial order on~$[n]$, an elimination forest~$F$ for a graph~$G=([n],E)$ has a set of linear extensions, defined as follows.
Given a forest~$F$ of rooted trees on~$[n]$, a \emph{linear extension} of~$F$ is a permutation~$\pi\in S_n$ such that for any two elements~$i,j\in [n]$, if $i$ is an ancestor of~$j$ in some tree of~$F$, then the value~$i$ is left of the value~$j$ in~$\pi$.
The linear extensions of an elimination forest~$F$ are therefore the elimination orderings that yield the forest~$F$; see Figure~\ref{fig:elim}.
The linear extensions of elimination forests for a graph~$G=([n],E)$ form classes of an \emph{equivalence relation on~$S_n$} that we denote by~$\equiv_G$.
%This relation is well-understood when the graph~$G$ is a path on $n$ vertices.
%In this case the equivalence classes are one-to-one with permutations of~$[n]$ avoiding the pattern~231, and they are counted by the Catalan numbers.
In order to generate the elimination forests for~$G$, we will later choose a set of representatives for the equivalence relation~$\equiv_G$, i.e., we will pick exactly one permutation from each equivalence class, and then generate this set of permutations.

\subsection{Zigzag languages and Algorithm~J}
\label{sec:algoJ}

We recap the most important definitions and results from~\cite{MR4391718}.
We write permutations~$\pi\in S_n$ in one-line notation as $\pi=\pi(1)\pi(2)\cdots \pi(n)=a_1a_2\cdots a_n$.
Moreover, we use $\ide_n=12\cdots n$ to denote the identity permutation, and $\varepsilon\in S_0$ to denote the empty permutation.
Algorithm~J, presented in~\cite{MR4391718} and shown below, is a simple greedy algorithm for generating a set of permutations~$L_n\seq S_n$.
The operation it uses to go from one permutation to the next is called a jump.
Given a permutation $\pi=a_1\cdots a_n$ with a substring $a_i\cdots a_j$ such that $a_i>a_{i+1},\ldots,a_j$, a \emph{right jump of the value~$a_i$ by $j-i$ steps} is a cyclic left rotation of this substring by one position to $a_{i+1}\cdots a_ja_i$.
Similarly, if $a_j>a_i,\ldots,a_{j-1}$, a \emph{left jump of the value~$a_j$ by $j-i$ steps} is a cyclic right rotation of this substring to $a_ja_i\cdots a_{j-1}$.
Given a set of permutations~$L_n\seq S_n$, a jump is \emph{minimal} (w.r.t.~$L_n$), if a jump of the same value in the same direction by fewer steps creates a permutation that is not in~$L_n$.

\begin{algo}{Algorithm~J}{Greedy minimal jumps}
This algorithm attempts to greedily generate a set of permutations $L_n\seq S_n$ using minimal jumps starting from an initial permutation $\pi_0 \in L_n$.
\begin{enumerate}[label={\bfseries J\arabic*.}, leftmargin=8mm, noitemsep, topsep=3pt plus 3pt]
\item{} [Initialize] Visit the initial permutation~$\pi_0$.
\item{} [Jump] Generate an unvisited permutation from~$L_n$ by performing a minimal jump of the largest possible value in the most recently visited permutation.
If no such jump exists, or the jump direction is ambiguous, then terminate.
Otherwise visit this permutation and repeat~J2.
\end{enumerate}
\end{algo}

For any permutation $\pi\in S_n$, we write $p(\pi)\in S_{n-1}$ for the permutation obtained by removing the largest entry~$n$.
Moreover, for any $\pi\in S_{n-1}$ and any $1\leq i\leq n$, we write $c_i(\pi)\in S_n$ for the permutation obtained by inserting the new largest element~$n$ at position~$i$ of~$\pi$.
A set of permutations~$L_n\seq S_n$ is called a \emph{zigzag language}, if either $n=0$ and $L_0=\{\varepsilon\}$, or if $n\geq 1$ and $L_{n-1}:=\{p(\pi) \mid \pi\in L_n\}$ is a zigzag language satisfying either one of the following conditions:
\begin{enumerate}[label={(z\arabic*)}, leftmargin=8mm, noitemsep, topsep=3pt plus 3pt]
\item For every $\pi\in L_{n-1}$ we have~$c_1(\pi)\in L_n$ and~$c_n(\pi)\in L_n$.
\item We have $L_n=\{c_n(\pi)\mid \pi\in L_{n-1}\}$.
\end{enumerate}
The next theorem asserts that Algorithm~J succeeds to generate zigzag languages.

\begin{theorem}[\cite{MR4344032}]
\label{thm:jump}
Given any zigzag language of permutations~$L_n$ and initial permutation $\pi_0=\ide_n$, Algorithm~J visits every permutation from~$L_n$ exactly once.
\end{theorem}

The reason we refer to~\cite{MR4344032} here instead of~\cite{MR4391718} is that the notion of zigzag language given in~\cite{MR4391718} omits condition~(z2) above.
However, for the present paper we need the more general definition with condition~(z2) introduced in~\cite{MR4344032} to be able to handle elimination forests for disconnected graphs~$G$.
Based on Theorem~\ref{thm:jump}, for any zigzag language~$L_n\seq S_n$, we write~$J(L_n)$ for the sequence of permutations from~$L_n$ generated by Algorithm~J with initial permutation~$\pi_0=\ide_n$.

It was shown in~\cite{MR4344032} that the sequence~$J(L_n)$ can be described recursively as follows.
For any $\pi\in L_{n-1}$ we let $\rvec{c}(\pi)$ be the sequence of all $c_i(\pi)\in L_n$ for $i=1,2,\ldots,n$, starting with~$c_1(\pi)$ and ending with~$c_n(\pi)$, and we let $\lvec{c}(\pi)$ denote the reverse sequence, i.e., it starts with~$c_n(\pi)$ and ends with~$c_1(\pi)$.
In words, those sequences are obtained by inserting into~$\pi$ the new largest value~$n$ from left to right, or from right to left, respectively, in all possible positions that yield a permutation from~$L_n$, skipping the positions that yield a permutation that is not in~$L_n$.
If $n=0$ then we have $J(L_0)=\varepsilon$, and if~$n\geq 1$ then we consider the sequence $J(L_{n-1})=:\pi_1,\pi_2,\ldots$ and we have
\begin{subequations}
\label{eq:JLn12}
\begin{equation}
\label{eq:JLn1}
  J(L_n)=\lvec{c}(\pi_1),\rvec{c}(\pi_2),\lvec{c}(\pi_3),\rvec{c}(\pi_4),\ldots
\end{equation}
if condition~(z1) holds, and we have
\begin{equation}
\label{eq:JLn2}
  J(L_n)=c_n(\pi_1),c_n(\pi_2),c_n(\pi_3),c_n(\pi_4),\ldots
\end{equation}
\end{subequations}
if condition~(z2) holds.

The alternating directions of jumps of the value~$n$ in~\eqref{eq:JLn1} motivate the name `zigzag' language.
It is easy to see that when run on the set $L_n=S_n$ of all permutations of~$[n]$, then the ordering~\eqref{eq:JLn1} uses only adjacent transpositions, and it is precisely the Steinhaus-Johnson-Trotter ordering; recall Table~\ref{tab:SJT} and Figure~\ref{fig:asso-perm}~(a) (condition~(z2) never holds in this case, so~\eqref{eq:JLn2} is irrelevant).

While Algorithm~J as stated requires bookkeeping of all previously visited permutations, and is as such not efficient, it can be made history-free and efficient; see Section~\ref{sec:algo}.

\section{Elimination forests for unit interval graphs and stars}
\label{sec:unit-stars}

Before proving our main result, Theorem~\ref{thm:jump-elim}, about chordal graphs in Section~\ref{sec:chordal} below, we first consider two special cases.
The first is that of filled graphs, as defined by Barnard and McConville~\cite{MR4176852}.
We show that they correspond to unit interval graphs, and apply a previous result from Hoang and M\"utze~\cite{MR4344032} on quotients of the weak order.
The second special case is that of stars.
We obtain new Gray codes for elimination trees for stars, which we show are in one-to-one correspondence with partial permutations.
This Gray code on partial permutations is the first of its kind.

\subsection{Unit interval graphs and quotients of the weak order}

Before describing our first result, we recall some standard order-theoretic definitions.

An \emph{inversion} in a permutation $\pi=a_1\cdots a_n\in S_n$ is a pair~$(a_i,a_j)$ with~$i<j$ and~$a_i>a_j$.
The \emph{weak order} of permutations in~$S_n$ is the containment order of their sets of inversions.
It is well-known that $S_n$, equipped with the weak order, is a lattice, i.e., joins $\pi\vee \rho$ and meets $\pi\wedge\rho$ are well-defined.

A \emph{lattice congruence} on a lattice~$L$ is an equivalence relation $\equiv$ that is compatible with taking joins and meets, i.e., if $x\equiv x'$ and $y \equiv y'$, then $x\vee y\equiv x'\vee y'$ and $x\wedge y\equiv x'\wedge y'$.
The lattice thus obtained on the equivalence classes of $\equiv$ is called the \emph{lattice quotient}~$L/\equiv$.

Recall from Section~\ref{sec:encoding} that $\equiv_G$ is the equivalence relation on~$S_n$ whose classes are the linear extensions of the elimination forests for~$G$.
Barnard and McConville considered the following problem: When is~$\equiv_G$ a lattice congruence of the weak order?
To this end, they call a graph~$G=([n],E)$ \emph{filled} if for each edge $\{i,k\}\in E$ with $i<k$ we have $\{i,j\}\in E$ and $\{j,k\}\in E$ for all $i<j<k$, and they show that this property is necessary and sufficient.

\begin{theorem}[\cite{MR4176852}]
\label{thm:filled}
The equivalence relation~$\equiv_G$ on~$S_n$ is a lattice congruence of the weak order if and only if $G$ is filled.
\end{theorem}

Pilaud and Santos~\cite{MR3964495} recently showed that the cover graph of any lattice quotient of the weak order on~$S_n$ can be realized as the skeleton of a polytope, called a \emph{quotientope}.
This family of polytopes simultaneously generalizes associahedra, permutahedra, hypercubes and many other known polytopes.
Another interpretation of Theorem~\ref{thm:filled} is therefore that if $G$ is filled, then the skeleton of its associahedron~$\cA(G)$ is that of a quotientope.
This allows us, in the special case of filled graphs, to directly reuse a recent result from Hoang and M\"utze~\cite{MR4344032}, who showed that Algorithm~J can be used to compute a Hamilton path on the skeleton of each quotientope.

\begin{theorem}[\cite{MR4344032}]
\label{thm:quot}
For any lattice congruence~$\equiv$ of the weak order on~$S_n$, there is a zigzag language~$R_n\seq S_n$ such that each equivalence class of~$\equiv$ contains exactly one permutation from~$R_n$.
Moreover, any two permutations from~$R_n$ that are visited consecutively by Algorithm~J are from equivalence classes that form a cover relation in the lattice quotient~$S_n/\equiv$.
\end{theorem}

Combining Theorem~\ref{thm:filled} and Theorem~\ref{thm:quot}, we obtain that the framework described in Section~\ref{sec:zigzag} applies directly in the special case of filled graphs.

\begin{corollary}
If a connected graph~$G=([n],E)$ is filled, then there is zigzag language~$R_n\seq S_n$ such that each equivalence class of~$\equiv_G$ contains exactly one permutation from~$R_n$.
Moreover, any two permutations from~$R_n$ that are visited consecutively by Algorithm~J are linear extensions of elimination forests that differ in a tree rotation.
\end{corollary}

The term `filled' is not standard in graph theory.
It is not difficult to check that filled graphs are chordal and claw-free.
We proceed to show that a graph~$G$ is filled if and only if it is a unit interval graph.
An \emph{intersection model} for a graph~$G=(V,E)$ is a bijection~$f$ between~$V$ and a collection~$S$ of sets such that $\{i,j\}\in E\Longleftrightarrow f(i)\cap f(j)\neq\emptyset$.
A graph is a \emph{(unit) interval graph} if it has an intersection model consisting of (unit) intervals of the real line.

\begin{lemma}
\label{lem:filled}
A graph~$G=([n],E)$ is filled if and only if it has an intersection model of unit intervals $\{[\ell_i,\ell_i+1]\mid i\in [n]\}$ such that $\ell_1<\ell_2<\cdots<\ell_n$ and no two interval endpoints coincide.
\end{lemma}

\begin{proof}
If there exists such a unit interval model for~$G$, then clearly $G$ is filled.
Now suppose that $G$ is filled.
We argue by induction on $n$ that the desired unit interval model for~$G$ exists.
The claim is trivial for $n=1$.
Suppose it holds for the subgraph of~$G$ induced by the first $n-1$ vertices, which is also filled.
If the vertex~$n$ is isolated in~$G$, we set $\ell_n:=\ell_{n-1}+2$ to obtain the desired unit interval model.
Otherwise we consider a neighbor~$i$ of the vertex~$n$.
As $G$ is filled, all vertices~$j$ such that $i<j<n$ must also be adjacent to both~$n$ and~$i$.
Hence the neighborhood of~$n$ is a clique induced by the vertices~$i, i+1,\ldots ,n-1$, where $i$ is the smallest neighbor of~$n$.
By the induction hypothesis, $G-n$ has a unit interval model, and it must be the case for this model that $\ell_{n-1}<\ell_i+1$ and that $\ell_{i-1}+1<\ell_i+1$, where we use the assumption that no two interval endpoints coincide.
We can therefore choose $\ell_n$ to lie in the open interval $(\max\{\ell_{n-1},\ell_{i-1}+1\},\ell_i+1)$ to obtain a unit interval model for~$G$ with the desired properties.
\end{proof}

\subsection{Stars and partial permutations}
\label{sec:stars}

Stars with at least four vertices are not unit interval graphs, and are therefore the next natural graph class to investigate.
As mentioned before, graph associahedra of stars are known as stellohedra.

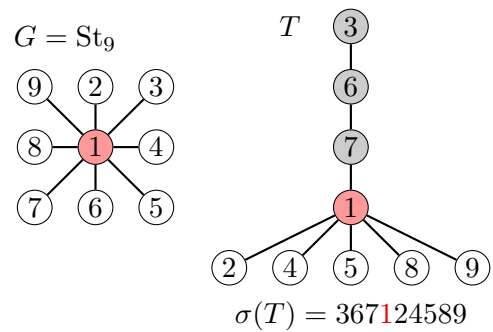
\begin{wrapfigure}{r}{0.54\textwidth}
\centering
\begin{tikzpicture}[scale=0.8,auto,swap]
\begin{scope}
  \foreach \pos/\name in {{(0,1)/2},{(1,1)/3},{(1,0)/4},{(1,-1)/5},{(0,-1)/6},{(-1,-1)/7},{(-1,0)/8},{(-1,1)/9}} \node[vertex] (\name) at \pos {$\name$};
  \node[vertex,fill=red!40] (1) at (0,0) {$1$};
  \foreach \dest in {2,3,4,5,6,7,8,9} \path[edge] (1) -- (\dest);
  \node at (-0.5,1.8) {$G=\Star_9$};
\end{scope}
\begin{scope}[xshift=4.2cm,yshift=-2cm]
  \foreach \pos/\name in {{(-2,0)/2},{(-1,0)/4},{(0,0)/5},{(1,0)/8},{(2,0)/9}} \node[vertex] (\name) at \pos {$\name$};
  \foreach \pos/\name in {{(0,2)/7},{(0,4)/3},{(0,3)/6}} \node[vertex,fill=black!20] (\name) at \pos {$\name$};
  \node[vertex,fill=red!40] (1) at (0,1) {$1$};
  \foreach \source/\dest in {3/6,6/7,7/1,1/2,1/4,1/5,1/8,1/9} \path[edge] (\source) -- (\dest);
  \node at (0,-0.8) {$\sigma(T)=367{\red 1}24589$};
  \node at (-1,4) {$T$};
\end{scope}
\end{tikzpicture}
\caption{An elimination tree~$T$ for~$\Star_9$.
The corresponding permutation is $\sigma(T)=367124589$, and the associated partial permutation is~$367-1=256$.}
\vspace{-1mm}
\label{fig:star}
\end{wrapfigure}
We take the vertex~1 as the center of the star, hence we consider the graph $\Star_n:=([n],E)$ with $E:=\{\{1,i\}\mid 2\leq i\leq n\}$ for~$n\geq 0$.
Elimination trees for $\Star_n$ are brooms; see Figure~\ref{fig:star}: a path composed of elements from a subset of~$[n]\setminus\{1\}$, followed by a subtree of height one rooted in~1.

A \emph{partial permutation} of~$[n]$ is a linearly ordered subset of~$[n]$.
The number of partial permutations of $[n]$ is $\sum_{k=0}^n n!/k!$ (OEIS~A000522).

Note that each elimination tree for~$\Star_n$, $n\geq 1$, corresponds to a partial permutation on~$[n-1]$, given by reading off the labels from the handle of the broom starting at the root and ending at the parent of~1, and subtracting 1 from those labels.
We see that elimination trees for $\Star_n$ are in one-to-one correspondence with partial permutations of~$[n-1]$.

In order to apply the framework from Section~\ref{sec:zigzag}, we define a natural mapping~$\sigma$ from the set~$\cE(\Star_n)$ of all elimination trees for the star to~$S_n$ as follows:
For an elimination tree~$T\in\cE(\Star_n)$, the permutation $\sigma(T)$ is obtained by reading off the labels from the handle of the broom starting at the root and ending at~1, followed by the remaining elements sorted increasingly.
For the special case of the empty elimination tree~$\emptyset$ obtained for the empty star~$\Star_0=\emptyset$, we define $\sigma(\emptyset):=\varepsilon$ to be the empty permutation.
Clearly, $\sigma(T)$ is a linear extension of~$T$.
We define $\Pi(\Star_n):=\{\sigma(T)\mid T\in\cE(\Star_n)\}$.

\begin{lemma}
The set of permutations~$\Pi(\Star_n)$ is a zigzag language.
\end{lemma}

\begin{table}[t!]
\centering
\renewcommand{\arraystretch}{0.9}
\setlength\tabcolsep{6pt}
\small
\begin{tabular}{|lllll|}
\hline
$n=1$ & $n=2$ & $n=3$ & $n=4$ & $n=5$ \\ \hline
$\varepsilon$ & $\varepsilon$ & $\varepsilon$ & $\varepsilon$ & $\varepsilon$ \\
&           &                     &                                & {\red 4} \\
&           &                     & {\orange 3}                    & {\red 4}{\orange 3} \\
&           &                     &                                & {\orange 3}{\red 4} \\
&           &                     &                                & {\orange 3} \\
&           & {\green 2}          & {\orange 3}{\green 2}          & {\orange 3}{\green 2} \\
&           &                     &                                & {\orange 3}{\green 2}{\red 4} \\
&           &                     &                                & {\orange 3}{\red 4}{\green 2} \\
&           &                     &                                & {\red 4}{\orange 3}{\green 2} \\
&           &                     & {\green 2}{\orange 3}          & {\red 4}{\green 2}{\orange 3} \\
&           &                     &                                & {\green 2}{\red 4}{\orange 3} \\
&           &                     &                                & {\green 2}{\orange 3}{\red 4} \\
&           &                     &                                & {\green 2}{\orange 3} \\
&           &                     & {\green 2}                     & {\green 2} \\
&           &                     &                                & {\green 2}{\red 4} \\
&           &                     &                                & {\red 4}{\green 2} \\
& {\blue 1} & {\green 2}{\blue 1} & {\green 2}{\blue 1}            & {\red 4}{\green 2}{\blue 1} \\
&           &                     &                                & {\green 2}{\red 4}{\blue 1} \\
&           &                     &                                & {\green 2}{\blue 1}{\red 4} \\
&           &                     &                                & {\green 2}{\blue 1} \\
&           &                     & {\green 2}{\blue 1}{\orange 3} & {\green 2}{\blue 1}{\orange 3} \\
&           &                     &                                & {\green 2}{\blue 1}{\orange 3}{\red 4} \\
&           &                     &                                & {\green 2}{\blue 1}{\red 4}{\orange 3} \\
&           &                     &                                & {\green 2}{\red 4}{\blue 1}{\orange 3} \\
&           &                     &                                & {\red 4}{\green 2}{\blue 1}{\orange 3} \\
&           &                     & {\green 2}{\orange 3}{\blue 1} & {\red 4}{\green 2}{\orange 3}{\blue 1} \\
&           &                     &                                & {\green 2}{\red 4}{\orange 3}{\blue 1} \\
&           &                     &                                & {\green 2}{\orange 3}{\red 4}{\blue 1} \\
&           &                     &                                & {\green 2}{\orange 3}{\blue 1}{\red 4} \\
&           &                     &                                & {\green 2}{\orange 3}{\blue 1} \\
&           &                     & {\orange 3}{\green 2}{\blue 1} & {\orange 3}{\green 2}{\blue 1} \\
&           &                     &                                & {\orange 3}{\green 2}{\blue 1}{\red 4} \\
&           &                     &                                & {\orange 3}{\green 2}{\red 4}{\blue 1} \\
&           &                     &                                & {\orange 3}{\red 4}{\green 2}{\blue 1} \\
&           &                     &                                & {\red 4}{\orange 3}{\green 2}{\blue 1} \\ \hline
\end{tabular}
\hspace{2mm}
\begin{tabular}{|lllll|}
\hline
$n=1$ & $n=2$ & $n=3$ & $n=4$ & $n=5$ \\ \hline
& & {\blue 1}{\green 2} & {\orange 3}{\blue 1}{\green 2} & {\red 4}{\orange 3}{\blue 1}{\green 2} \\
& &                     &                                & {\orange 3}{\red 4}{\blue 1}{\green 2} \\
& &                     &                                & {\orange 3}{\blue 1}{\red 4}{\green 2} \\
& &                     &                                & {\orange 3}{\blue 1}{\green 2}{\red 4} \\
& &                     &                                & {\orange 3}{\blue 1}{\green 2} \\
& &                     & {\blue 1}{\orange 3}{\green 2} & {\blue 1}{\orange 3}{\green 2} \\
& &                     &                                & {\blue 1}{\orange 3}{\green 2}{\red 4} \\
& &                     &                                & {\blue 1}{\orange 3}{\red 4}{\green 2} \\
& &                     &                                & {\blue 1}{\red 4}{\orange 3}{\green 2} \\
& &                     &                                & {\red 4}{\blue 1}{\orange 3}{\green 2} \\
& &                     & {\blue 1}{\green 2}{\orange 3} & {\red 4}{\blue 1}{\green 2}{\orange 3} \\
& &                     &                                & {\blue 1}{\red 4}{\green 2}{\orange 3} \\
& &                     &                                & {\blue 1}{\green 2}{\red 4}{\orange 3} \\
& &                     &                                & {\blue 1}{\green 2}{\orange 3}{\red 4} \\
& &                     &                                & {\blue 1}{\green 2}{\orange 3} \\
& &                     & {\blue 1}{\green 2}            & {\blue 1}{\green 2} \\
& &                     &                                & {\blue 1}{\green 2}{\red 4} \\
& &                     &                                & {\blue 1}{\red 4}{\green 2} \\
& &                     &                                & {\red 4}{\blue 1}{\green 2} \\
& & {\blue 1}           & {\blue 1}                      & {\red 4}{\blue 1} \\
& &                     &                                & {\blue 1}{\red 4} \\
& &                     &                                & {\blue 1} \\
& &                     & {\blue 1}{\orange 3}           & {\blue 1}{\orange 3} \\
& &                     &                                & {\blue 1}{\orange 3}{\red 4} \\
& &                     &                                & {\blue 1}{\red 4}{\orange 3} \\
& &                     &                                & {\red 4}{\blue 1}{\orange 3} \\
& &                     & {\orange 3}{\blue 1}           & {\red 4}{\orange 3}{\blue 1} \\
& &                     &                                & {\orange 3}{\red 4}{\blue 1} \\
& &                     &                                & {\orange 3}{\blue 1}{\red 4} \\
& &                     &                                & {\orange 3}{\blue 1} \\ \hline
\multicolumn{5}{c}{} \\
\multicolumn{5}{c}{} \\
\multicolumn{5}{c}{} \\
\multicolumn{5}{c}{} \\
\multicolumn{5}{c}{} \\
\end{tabular}
\vspace{2mm}
\caption{The Gray code on partial permutations obtained from Algorithm~J.
The parameter~$n$ refers to the number of vertices of the star, and the partial permutations are on the set~$[n-1]$.
}
\label{tab:partial}
\end{table}

\begin{proof}
We use the definition of zigzag languages from Section~\ref{sec:algoJ}, and we argue by induction on~$n$.
For the base case of the induction $n=0$ we have $\Pi(\Star_0)=\{\varepsilon\}$, which is a zigzag language.

For the induction step let $n\geq 1$ and assume by induction that $\Pi(\Star_{n-1})$ is a zigzag language.
First observe that removing that largest entry~$n$ in a permutation $\pi\in\Pi(\Star_n)$ yields a permutation in~$\Pi(\Star_{n-1})$, and every permutation in~$\Pi(\Star_{n-1})$ is obtained in this way.
Specifically, if $T$ is the elimination tree such that $\sigma(T)=\pi$, then $p(\pi)=\sigma(T')$ for the elimination tree~$T'$ for~$\Star_{n-1}$ obtained by removing vertices as described by~$T$, but ignoring the leaf~$n$ of~$\Star_n$.
It follows that $\{p(\pi)\mid \pi\in\Pi(\Star_n)\}=\Pi(\Star_{n-1})$.
Moreover, for any $\pi\in\Pi(\Star_{n-1})$, both $c_1(\pi)$ and $c_n(\pi)$ belong to $\Pi(\Star_n)$, so condition~(z1) is satisfied.
Specifically, consider the elimination tree~$T$ with $\sigma(T)=\pi$.
Then $c_1(\pi)=\sigma(T')$ for the elimination tree~$T'$ obtained with the elimination ordering~$c_1(\pi)$, removing the leaf~$n$ of the star first.
Similarly, $c_n(\pi)=\sigma(T'')$ for the elimination tree~$T''$ obtained with the elimination ordering $c_n(\pi)$, removing the leaf~$n$ of the star last.
\end{proof}

\begin{figure}[t!]
\makebox[0cm]{ % artificial box to center the picture
\includegraphics[page=3]{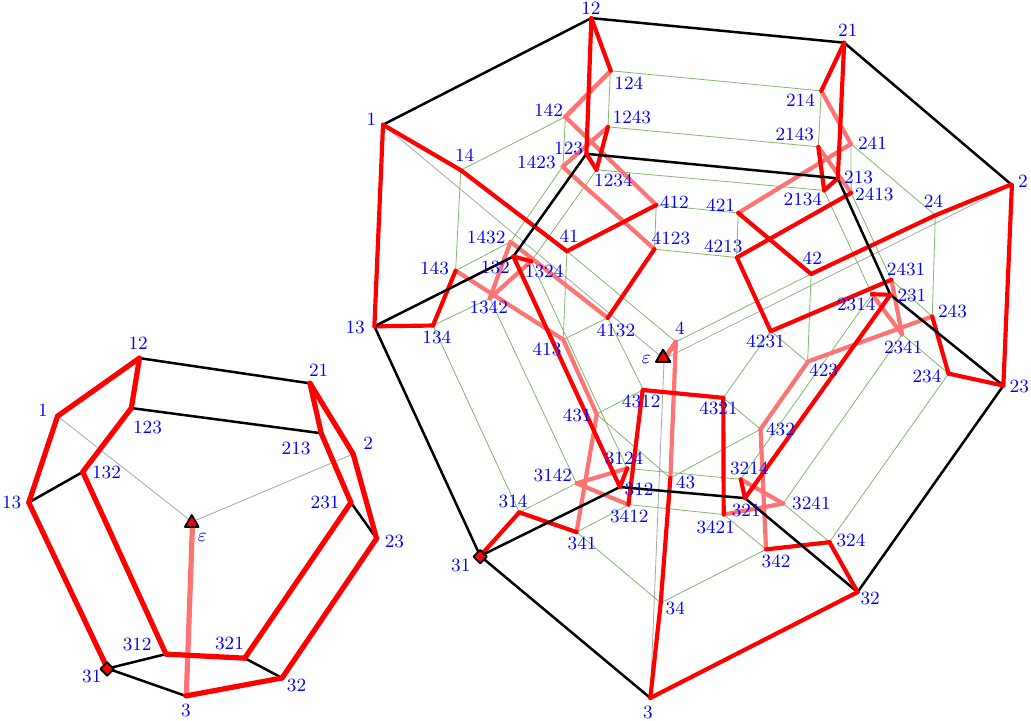}
}
\caption{The three- and four-dimensional stellohedron, and the Hamilton paths obtained from Algorithm~J.
The four-dimensional stellohedron is visualized by its Schlegel diagram.}
\label{fig:stello}
\end{figure}

Algorithm~J therefore applies again.
In order to state our result in terms of partial permutations, it is interesting to understand the effect of a rotation in an elimination tree on the partial permutation represented by this tree.
If the rotation does not involve the center~$1$ of the star, then it corresponds to an adjacent transposition in the partial permutation.
On the other hand, if the rotation does involve the center, then it corresponds to either deleting the trailing element of the partial permutation, or appending a new one.
Note that rotations on the elimination trees for~$\Star_n$ are also in one-to-one correspondence with minimal jumps in permutations in~$\Pi(\Star_n)$.
Specifically, if the rotation does not involve the center~$1$ of the star, then the corresponding adjacent transposition is a jump by 1~step, which is clearly minimal.
On the other hand, if the rotation does involve the center, then the definition of~$\sigma(T)$ requires that vertices of an elemination tree~$T\in\cE(\Star_n)$ not belonging to the handle of the broom are sorted increasingly, so there is only a single valid location for every non-handle vertex in the permutation~$\sigma(T)$, meaning that the corresponding jump to that location must be minimal.

We summarize this as follows.

\begin{lemma}
The following three sets of operations are in one-to-one correspondence:
\begin{itemize}[leftmargin=5mm, noitemsep, topsep=3pt plus 3pt]
\item rotations in elimination trees for~$\Star_n$,
\item minimal jumps in permutations in $\Pi(\Star_n)$,
\item adjacent transpositions, and deletions and insertions of a trailing element in partial permutations of~$[n-1]$.
\end{itemize}
\end{lemma}

We therefore obtain Gray codes for partial permutations, in which any two consecutive partial permutations differ only by one such operation.

\begin{corollary}
\label{cor:partial}
Algorithm~J generates a Gray code on partial permutations of~$[n-1]$, in which any two consecutive partial permutations differ either by an adjacent transposition, or a deletion or insertion of a trailing element.
\end{corollary}

These Gray codes for $n\leq 5$ are given in Table~\ref{tab:partial}.
The Hamilton paths obtained on the three- and four-dimensional stellohedron are illustrated in Figure~\ref{fig:stello}.

\section{Elimination forests for chordal graphs}
\label{sec:chordal}

We now derive Algorithm~R for generating all elimination forests for a chordal graph by tree rotations.
Unit interval graphs and stars are also chordal, therefore everything we say will be a generalization of the two special cases discussed in the previous section.
We first analyze properties of elimination forests for chordal graphs, then proceed to define a zigzag language of permutations for them, and we then apply the framework from Section~\ref{sec:zigzag}.

\subsection{Chordal graph basics}

Recall that a graph is \emph{chordal} if every induced cycle has length three.
Also recall that given a graph~$G=(V,E)$, a \emph{perfect elimination ordering}, or PEO for short, is a linear ordering of the vertices in~$V$ such that every vertex~$x$ induces a clique together with its neighbors in~$G$ that come before~$x$ in the ordering.
A useful characterization of chordal graphs is given by the following result due to Fulkerson and Gross.

\begin{lemma}[\cite{MR186421}]
\label{lem:chordal-PEO}
A graph~$G$ is chordal if and only if it has a PEO.
\end{lemma}

In what follows, we consider a chordal graph~$G=([n],E)$, where the ordering $1,2,\ldots,n$ is a PEO of~$G$.
We then say that $G=([n],E)$ is a PEO graph.
Clearly, if $G=([n],E)$ is a PEO graph, then the subgraph~$G^{[\nu]}$ of~$G$ induced by the first $\nu$ vertices is also a PEO graph for all $\nu=0,\ldots,n$.
We establish the following characterization of PEOs.

\begin{lemma}
\label{lem:smaller-child}
A graph~$G=([n],E)$ is a PEO graph if and only if for all elimination forests for~$G$, every vertex $j\in [n]$ has at most one child that is smaller than~$j$.
\end{lemma}

To prove this lemma and others later, we will use the following two auxiliary lemmas.
A vertex of a graph~$G=([n],E)$ is called \emph{simplicial}, if it induces a clique together with its smaller neighbors in~$G$.
By definition, $G$ is a PEO graph if and only if all of its vertices are simplicial.

\begin{lemma}
\label{lem:no-peak}
Let $G=([n],E)$ be a graph with a simplicial vertex~$j\in[n]$, and let $H$ be an induced subgraph of~$G$ containing~$j$.
Then along any shortest path between two vertices of~$H$, there is no triple of vertices~$i,j,k$ in that order with $j>i,k$.
\end{lemma}

\begin{proof}
If there was such a triple, then $j$ is assumed to be simplicial and $j>i,k$, the edge~$\{i,k\}$ must be in~$H$, so the path is not a shortest path, a contradiction.
\end{proof}

\begin{lemma}
\label{lem:smaller-child-j}
Let $G=([n],E)$ be a graph with a simplicial vertex~$j\in[n]$.
Then in every elimination forest for~$G$, the vertex~$j$ has at most one child that is smaller than~$j$.
\end{lemma}

\begin{proof}
Suppose for the sake of contradiction that there is an elimination forest for~$G$ with an elimination tree~$T$ containing the vertex~$j$ and two children~$i$ and~$k$ with $j>i,k$.
Let~$U$ be the set of ancestors of~$j$ in~$T$.
This means that the graph~$G-U$ contains a component~$H$ with the vertex~$j$, and removing~$j$ from~$H$ leaves at least two smaller components, one containing~$i$ and one containing~$k$.
Let $i'$ and $k'$ be the neighbors of~$j$ that lie on a shortest path from~$i$ to~$k$ in~$H$.
By Lemma~\ref{lem:no-peak} we have~$i'<j$ and $k'<j$, and there is no edge between~$i'$ and~$k'$ in~$H$, contradicting the assumption that~$j$ is simplicial.
\end{proof}

\begin{proof}[Proof of Lemma~\ref{lem:smaller-child}]
First suppose that $G=([n],E)$ is a PEO graph.
As every vertex~$j\in[n]$ is simplicial in~$G$, Lemma~\ref{lem:smaller-child-j} yields that in every elimination forest for~$G$, every vertex~$j\in[n]$ has at most one child that is smaller than~$j$.

To prove the other direction, suppose that $G=([n],E)$ is not a PEO graph.
Then there are vertices $i,j,k$ with $j>i,k$ and $\{j,i\},\{j,k\}\in E$ and $\{i,k\}\notin E$.
We consider the elimination forest that arises from the following elimination order:
We first remove all vertices from $[n]\setminus \{i,j,k\}$ in arbitrary order, and then the vertices $j,i,k$ in that order.
Clearly, this creates an elimination tree in which the vertex~$j$ has precisely the two smaller children~$i$ and~$k$.
\end{proof}

\subsection{Deletion and insertion in elimination forests}
\label{sec:del-ins}

Given a PEO graph~$G=([n],E)$, Lemma~\ref{lem:smaller-child} implies that vertex~$n$ has at most one child in any elimination forest~$F$ for~$G$.
This leads us to define natural deletion and insertion operations on~$F$, which will allow us to recursively generate all elimination forests~$\cE(G)$ by rotations.

We first define \emph{deletion}.
Specifically, the forest $p(F)$ on the vertex set~$[n-1]$ is obtained from~$F$ as follows; see Figure~\ref{fig:delete}.
\begin{enumerate}[label=(\alph*),leftmargin=6mm, noitemsep, topsep=3pt plus 3pt]
\item If $n$ is the root of a tree~$T$ with at least two vertices in~$F$, then by Lemma~\ref{lem:smaller-child} it has exactly one child.
Then $p(F)$ is obtained by removing~$n$ from~$T$ and making its only child the new root.
\item If $n$ has both a parent and a child in a tree~$T$ of~$F$, then by Lemma~\ref{lem:smaller-child} it has exactly one child.
Then $p(F)$ is obtained by removing $n$ from~$T$ and connecting its parent to its child.
\item If $n$ is a leaf of a tree in~$F$ or an isolated vertex, then $p(F)$ is obtained by removing~$n$.
\end{enumerate}

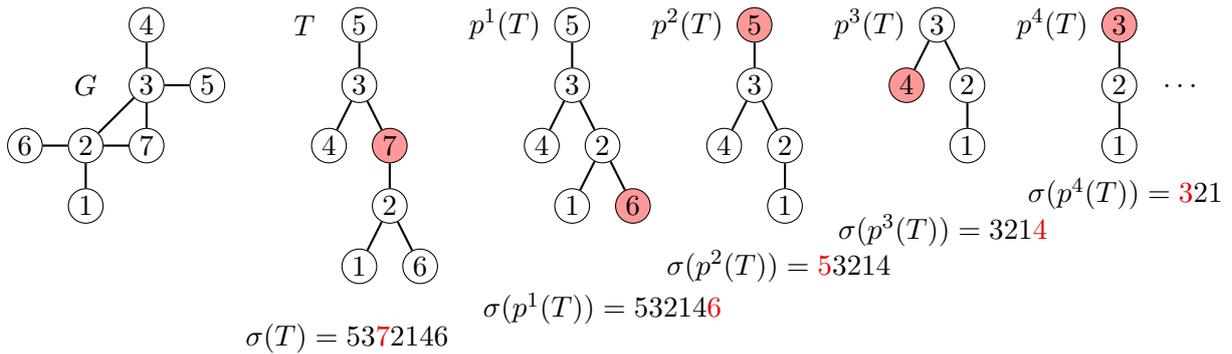
\begin{figure}[h!]
\centering
\begin{tikzpicture}[scale=.8,auto,swap]
  \foreach \pos/\name in {{(1,2)/6}, {(2,1)/1}, {(3,2)/7}, {(3,4)/4}, {(4,3)/5}, {(2,2)/2}, {(3,3)/3}} \node[vertex] (\name) at \pos {$\name$};
  \foreach \source/\dest in {6/2,1/2,2/3,2/7,7/3,3/4,3/5} \path[edge] (\source) -- (\dest);
  \node at (2,3) {$G$};

\begin{scope}[xshift=4.5cm,yshift=-1cm]
  \foreach \pos/\name in {{(2,1)/1},{(2.5,2)/2},{(2,4)/3},{(1.5,3)/4},{(2,5)/5},{(3,1)/6}} \node[vertex] (\name) at \pos {$\name$};
  \node[vertex, fill=red!40] (7) at (2.5,3) {$7$};
  \foreach \source/\dest in {5/3,3/4,3/7,7/2,2/1,2/6} \path[edge] (\source) -- (\dest);
  \node at (1.1,5) {$T$};
  \node at (1.8,-0.2) {$\sigma(T)=53{\red 7}2146$};
\end{scope}

\begin{scope}[xshift=8cm,yshift=-1cm]
  \foreach \pos/\name in {{(2,2)/1},{(2.5,3)/2},{(2,4)/3},{(1.5,3)/4},{(2,5)/5}} \node[vertex] (\name) at \pos {$\name$};
  \node[vertex, fill=red!40] (6) at (3,2) {$6$};
  \foreach \source/\dest in {5/3,3/4,3/2,2/1,2/6} \path[edge] (\source) -- (\dest);
  \node at (0.9,5) {$p^1(T)$};
  \node at (2.5,0.3) {$\sigma(p^1(T))=\linebreak 53214{\red 6}$};
\end{scope}

\begin{scope}[xshift=11cm,yshift=-1cm]
  \foreach \pos/\name in {{(2.5,2)/1},{(2.5,3)/2},{(2,4)/3},{(1.5,3)/4}} \node[vertex] (\name) at \pos {$\name$};
  \node[vertex, fill=red!40] (5) at (2,5) {$5$};
  \foreach \source/\dest in {5/3,3/4,3/2,2/1} \path[edge] (\source) -- (\dest);
  \node at (0.9,5) {$p^2(T)$};
  \node at (2.4,1) {$\sigma(p^2(T))={\red 5}3214$};
\end{scope}

\begin{scope}[xshift=14cm,yshift=-1cm]
  \foreach \pos/\name in {{(2.5,3)/1},{(2.5,4)/2},{(2,5)/3}} \node[vertex] (\name) at \pos {$\name$};
  \node[vertex, fill=red!40] (4) at (1.5,4) {$4$};
  \foreach \source/\dest in {2/1,3/2,3/4} \path[edge] (\source) -- (\dest);
  \node at (0.9,5) {$p^3(T)$};
  \node at (2.1,1.6) {$\sigma(p^3(T))=321{\red 4}$};
\end{scope}

\begin{scope}[xshift=17cm,yshift=-1cm]
  \foreach \pos/\name in {{(2,3)/1},{(2,4)/2}} \node[vertex] (\name) at \pos {$\name$};
  \node[vertex, fill=red!40] (3) at (2,5) {$3$};
  \foreach \source/\dest in {2/1,3/2} \path[edge] (\source) -- (\dest);
  \node at (0.9,5) {$p^4(T)$};
  \node at (2.1,2.2) {$\sigma(p^4(T))={\red 3}21$};
  \node at (3,4) {$\ldots$};
\end{scope}
\end{tikzpicture}
\caption{Repeated deletions in an elimination tree~$T$ for~$G$, and the corresponding permutations.}
\label{fig:delete}
\end{figure}

\begin{lemma}
\label{lem:delete}
The forest~$p(F)$ is an elimination forest for~$G-n$.
Furthermore, for every elimination forest~$F'$ for~$G-n$, there exists an elimination forest~$F$ for~$G$ such that $F'=p(F)$.
\end{lemma}

\begin{proof}
Consider the elimination forest~$F$ for~$G$, and consider the tree~$T$ in~$F$ containing~$n$.
Let $x_1,\ldots,x_k$ be the sequence of ancestors of~$n$ in~$T$, starting at the root~$x_1$ of~$T$ and ending at the parent~$x_k$ of~$n$.
For every $i=1,\ldots,k$, consider the component~$H$ of~$G-\{x_1,\ldots,x_{i-1}\}$ containing~$x_i,\ldots,x_k$ and~$n$.
By removing~$x_i$ from~$H$, this component is split into smaller components~$H_1,\ldots,H_\ell$, one of them, $H_1$ say, containing~$x_{i+1},\ldots,x_k$ and~$n$ ($\ell$ is the number of children of~$x_i$ in~$T$).
As the neighborhood of~$n$ in~$G$ is a clique, $H-n$ is one connected component of~$(G-n)-\{x_1,\ldots,x_{i-1}\}$, and removing~$x_i$ from~$H-n$ creates precisely the components~$H_1-n,H_2,\ldots,H_\ell$.

If $n$ is neither an isolated vertex in~$F$ nor a leaf of~$T$, then let~$i$ be the unique child of~$n$ in~$T$ and consider the component~$H$ of~$G-\{x_1,\ldots,x_k\}$ containing~$n$ and~$i$.
We know that $H-n$ is one connected component containing~$i$, which has as an elimination tree the subtree of~$T$ rooted at~$i$.
Consequently, the subtree of~$T$ rooted at~$i$ is an elimination tree for the connected component~$H-n$ of~$(G-n)-\{x_1,\ldots,x_k\}$.

Combining the previous two observations shows that~$p(F)$ is indeed an elimination forest for~$G-n$, which proves the first part of the lemma.

To prove the second part of the lemma, let $F'$ be an elimination forest for~$G-n$.
If $n$ is an isolated vertex of~$G$, then $F$ is obtained from~$F'$ by adding the isolated vertex~$n$, and clearly we have $F'=p(F)$.
On the other hand, if $n$ is not isolated in~$G$, then let~$T$ be the tree of~$F'$ containing the neighbors of~$n$ in~$G$.
Then $F$ is obtained from~$F'$ by making $n$ the new root of~$T$ (in the corresponding elimination ordering for~$G$, we first remove~$n$ to obtain~$G-n$).
\end{proof}

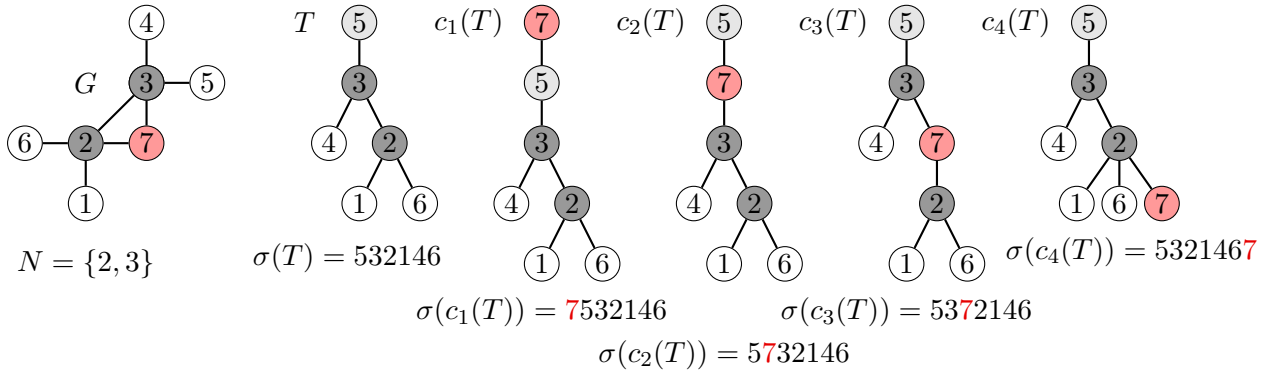
\begin{figure}[b!]
\centering
\begin{tikzpicture}[scale=.8,auto,swap]
  \foreach \pos/\name in {{(1,2)/6}, {(2,1)/1}, {(3,4)/4}, {(4,3)/5}, {(3,3)/3}} \node[vertex] (\name) at \pos {$\name$};
  \node[vertex, fill=black!40] (2) at (2,2) {$2$};
  \node[vertex, fill=black!40] (3) at (3,3) {$3$};
  \node[vertex, fill=red!40] (7) at (3,2) {$7$};
  \foreach \source/\dest in {6/2,1/2,2/3,2/7,7/3,3/4,3/5} \path[edge] (\source) -- (\dest);
  \node at (2,3) {$G$};
  \node at (2,0) {$N=\{2,3\}$};

\begin{scope}[xshift=4.5cm,yshift=0cm]
  \foreach \pos/\name in {{(2,1)/1},{(1.5,2)/4},{(3,1)/6}} \node[vertex] (\name) at \pos {$\name$};
  \node[vertex, fill=black!10] (5) at (2,4) {$5$};
  \node[vertex, fill=black!40] (2) at (2.5,2) {$2$};
  \node[vertex, fill=black!40] (3) at (2,3) {$3$};
  \foreach \source/\dest in {5/3,3/4,3/2,2/1,2/6} \path[edge] (\source) -- (\dest);
  \node at (1.1,4) {$T$};
  \node at (1.8,0.1) {$\sigma(T)=532146$};
\end{scope}

\begin{scope}[xshift=7.5cm,yshift=-1cm]
  \foreach \pos/\name in {{(2,1)/1},{(1.5,2)/4},{(3,1)/6}} \node[vertex] (\name) at \pos {$\name$};
  \node[vertex, fill=black!10] (5) at (2,4) {$5$};
  \node[vertex, fill=black!40] (2) at (2.5,2) {$2$};
  \node[vertex, fill=black!40] (3) at (2,3) {$3$};
  \node[vertex, fill=red!40] (7) at (2,5) {$7$};
  \foreach \source/\dest in {5/3,3/4,5/7,3/2,2/1,2/6} \path[edge] (\source) -- (\dest);
  \node at (0.8,5) {$c_1(T)$};
  \node at (2,0.2) {$\sigma(c_1(T))={\red 7}532146$};
\end{scope}

\begin{scope}[xshift=10.5cm,yshift=-1cm]
  \foreach \pos/\name in {{(2,1)/1},{(1.5,2)/4},{(3,1)/6}} \node[vertex] (\name) at \pos {$\name$};
  \node[vertex, fill=black!10] (5) at (2,5) {$5$};
  \node[vertex, fill=black!40] (2) at (2.5,2) {$2$};
  \node[vertex, fill=black!40] (3) at (2,3) {$3$};
  \node[vertex, fill=red!40] (7) at (2,4) {$7$};
  \foreach \source/\dest in {7/3,3/4,5/7,3/2,2/1,2/6} \path[edge] (\source) -- (\dest);
  \node at (0.8,5) {$c_2(T)$};
  \node at (2,-0.5) {$\sigma(c_2(T))=5{\red 7}32146$};
\end{scope}

\begin{scope}[xshift=13.5cm,yshift=-1cm]
  \foreach \pos/\name in {{(2,1)/1},{(1.5,3)/4},{(3,1)/6}} \node[vertex] (\name) at \pos {$\name$};
  \node[vertex, fill=black!10] (5) at (2,5) {$5$};
  \node[vertex, fill=black!40] (2) at (2.5,2) {$2$};
  \node[vertex, fill=black!40] (3) at (2,4) {$3$};
  \node[vertex, fill=red!40] (7) at (2.5,3) {$7$};
  \foreach \source/\dest in {5/3,3/4,3/7,7/2,2/1,2/6} \path[edge] (\source) -- (\dest);
  \node at (0.8,5) {$c_3(T)$};
  \node at (2,0.2) {$\sigma(c_3(T))=53{\red 7}2146$};
\end{scope}

\begin{scope}[xshift=16.5cm,yshift=0cm]
  \foreach \pos/\name in {{(1.8,1)/1},{(1.5,2)/4},{(2.5,1)/6}} \node[vertex] (\name) at \pos {$\name$};
  \node[vertex, fill=black!10] (5) at (2,4) {$5$};
  \node[vertex, fill=black!40] (2) at (2.5,2) {$2$};
  \node[vertex, fill=black!40] (3) at (2,3) {$3$};
  \node[vertex, fill=red!40] (7) at (3.2,1) {$7$};
  \foreach \source/\dest in {5/3,3/4,3/2,2/1,2/6,2/7} \path[edge] (\source) -- (\dest);
  \node at (0.8,4) {$c_4(T)$};
  \node at (2.7,0.2) {$\sigma(c_4(T))=532146{\red 7}$};
\end{scope}
\end{tikzpicture}
\caption{Insertion in an elimination tree~$T$ for $G-n$, and the corresponding permutations.
The neighbors $N=\{2,3\}$ of~7 in~$G$ are darkgray.
We have $\lambda=3$ and the insertion path in~$T$ is $(x_1,x_2,x_3)=(5,3,2)$, where the vertices on this path that are not in~$N$ are lightgray.
All elimination trees $c_i(T)$, $i=1,\ldots,\lambda+1$, for~$G$ are shown from left to right.}
\label{fig:insert}
\end{figure}

We now define \emph{insertion}; see Figure~\ref{fig:insert}.
We distinguish two cases.
If $n$ is an isolated vertex in~$G$, then for any elimination forest~$F$ of~$G-n$ we define~$\lambda:=0$ and we let $c_1(F)$ be the forest obtained as the disjoint union of~$F$ with the (isolated) vertex~$n$.
If $n$ is not isolated in~$G$, then we consider the set~$N$ of neighbors of~$n$ in~$G$, and for any elimination forest~$F$ for $G-n$, we consider the tree~$T$ in~$F$ that contains the vertices of~$N$.
As $N$ is a clique in~$G$, any two of these vertices are in an ancestor-descendant relation in~$T$.
We let $(x_1,\ldots,x_\lambda)$, $\lambda=\lambda(F)$, be the path in~$T$ starting at the root and ending at the deepest node from~$N$, which we refer to as \emph{insertion path}.
Clearly, we have $N\seq \{x_1,\ldots,x_\lambda\}$ and $x_\lambda\in N$, but there may also be vertices $x_i\notin N$.
For any $1\leq i\leq \lambda+1$ we define $c_i(F)$ as follows:
\begin{enumerate}[label=(\alph*),leftmargin=6mm, noitemsep, topsep=3pt plus 3pt]
\item If $i=1$, then $c_i(F)$ is obtained from~$F$ by making~$n$ the new root of~$T$.
\item If $2\leq i\leq \lambda$, then $c_i(F)$ is obtained from~$F$ by subdividing the edge between~$x_{i-1}$ and~$x_i$ of~$T$ by the vertex~$n$.
\item If $i=\lambda+1$, then $c_i(F)$ is obtained from~$F$ by making~$n$ a leaf of~$x_\lambda$ in~$T$.
\end{enumerate}

The following lemma asserts that the operations of deletion and insertion in elimination trees are inverse to each other, and that in the sequence of elimination forests $c_i(F)$, $i=1,\ldots,\lambda+1$, any two consecutive forests differ in a tree rotation.

\begin{lemma}
\label{lem:insert}
For any $1\leq i\leq \lambda+1$, the forest~$c_i(F)$ is an elimination forest for~$G$ and we have $p(c_i(F))=F$, and for every elimination forest~$F'$ for~$G$ with $p(F')=F$ there is an index~$i$ such that $c_i(F)=F'$.
Moreover, the vertex~$n$ is in depth~$i-1$ in~$c_i(F)$, and $c_i(F)$ and~$c_{i+1}(F)$ differ in a rotation of the edge~$\{x_i,n\}$.
\end{lemma}

\begin{proof}
The elimination orderings that yield~$c_i(F)$ are as follows:
If $i=1$, we first remove $n$, and then the remaining vertices as given by~$F$.
If $i=\lambda+1$, we remove the vertices as given by~$F$, and the vertex~$n$ last.
In this case $n$ becomes a child of~$x_\lambda$ in the corresponding elimination tree, as $n$ becomes an isolated vertex after removal of~$x_\lambda$.
If $2\leq i\leq \lambda$, we remove the vertex~$n$ after~$x_{i-1}$ and before~$x_i$.
\end{proof}

For any PEO graph~$G=([n],E)$, Lemmas~\ref{lem:delete} and~\ref{lem:insert} give rise to a tree~$\fT(G)$ on the sets of elimination forests~$\cE(G^{[\nu]})$ for $\nu=0,\ldots,n$; see Figure~\ref{fig:tree}.
Specifically, the tree~$\fT(G)$ has the empty graph~$\emptyset$ as a root, and for $\nu=1,\ldots,n$, the set of children of any elimination forest~$F\in\cE(G^{[\nu-1]})$ is precisely the set $\{c_i(F)\in\cE(G^{[\nu]})\mid 1\leq i\leq \lambda(F)+1\}$.
Conversely, the parent of each~$F\in\cE(G^{[\nu]})$, $\nu=1,\ldots,n$, is $p(F)\in\cE(G^{[\nu-1]})$.
By Lemma~\ref{lem:insert}, the set of nodes in depth~$\nu$ of the tree is precisely the set~$\cE(G^{[\nu]})$.
This tree is unordered, i.e., there is no ordering associated with the children of each elimination forest.
However, we will use Algorithm~J to generate the set of all elimination forests~$\cE(G^{[\nu]})$ for $\nu=0,\ldots,n$ in a particular order, and this will equip the unordered tree with an ordering of the nodes in each level; see Figure~\ref{fig:Jtree}.

\begin{figure}
\centering
\includegraphics[page=1]{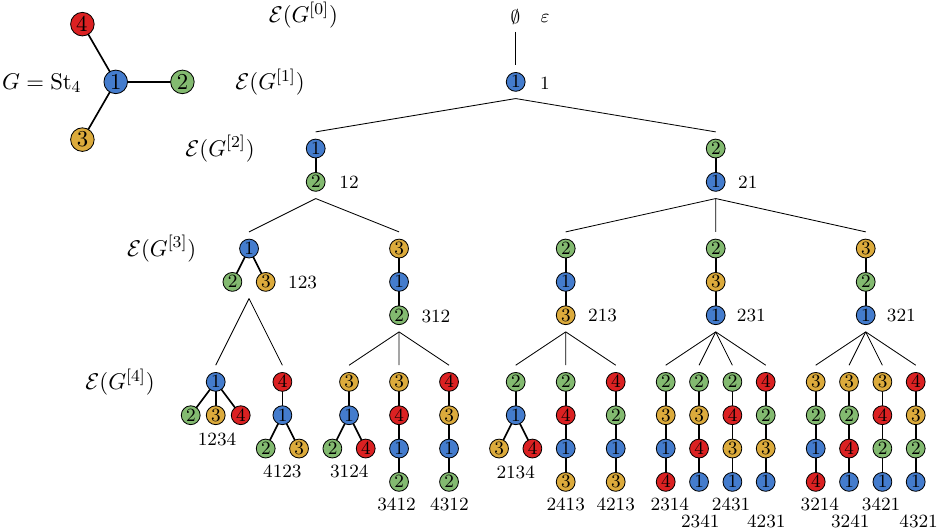}
\caption{The tree~$\fT(G)$ of all elimination trees for~$G^{[\nu]}$ for $\nu=0,\ldots,n$ for $G=\Star_4$.
Next to each elimination tree~$T$ is the permutation~$\sigma(T)$, which is a linear extension of~$T$.}
\label{fig:tree}
\end{figure}

\subsection{Proof of Theorem~\ref{thm:jump-elim}}
\label{sec:zigzag-chordal}

Generalizing the approach from Section~\ref{sec:stars}, we now define a mapping~$\sigma:\cE(G)\rightarrow S_n$ from the set of all elimination forests for a chordal graph~$G$ to permutations of~$[n]$.
The recursive definition of this mapping will facilitate the proof that the set of permutations thus obtained is a zigzag language.

Given a PEO graph~$G=([n],E)$ and an elimination forest~$F$ for~$G$, the permutation $\sigma(F)\in S_n$ is defined recursively as follows; see Figures~\ref{fig:delete} and~\ref{fig:tree}:
If $n=0$, then for the empty elimination forest~$\emptyset$ we define $\sigma(\emptyset):=\varepsilon$, and if $n\geq 1$ we define $\pi:=\sigma(p(F))\in S_{n-1}$ and consider three cases:
\begin{enumerate}[label=(\alph*),leftmargin=6mm, noitemsep, topsep=3pt plus 3pt]
\item If $n$ is the root of a tree with at least two vertices in~$F$, then $\sigma(F):=c_1(\pi)$.
\item If $n$ has both a parent and a child~$i$ in a tree of~$F$ (recall Lemma~\ref{lem:smaller-child}), then $\sigma(F)$ is obtained by inserting $n$ right before~$i$ in~$\pi$.
\item If $n$ is a leaf of a tree in~$F$ or an isolated vertex, then $\sigma(F)=c_n(\pi)$.
\end{enumerate}
We also define $\Pi(G):=\{\sigma(F)\mid F\in\cE(G)\}$.
Observe that $\sigma(F)$ is a linear extension of~$F$.
In particular, $\sigma:\cE(G)\rightarrow \Pi(G)$ is a bijection.
Also note that the deletion operation~$p$ on elimination forests introduced in the previous section and the mapping~$\sigma$ commute, i.e., we have $\sigma(p(F))=p(\sigma(F))$ for all~$F\in\cE(G)$.

\begin{lemma}
\label{lem:chordal-zigzag}
For every PEO graph~$G=([n],E)$, the set~$\Pi(G)\seq S_n$ is a zigzag language.
\end{lemma}

\begin{proof}
We define $L_\nu:=\Pi(G^{[\nu]})$ for $\nu=0,\ldots,n$, and we argue by induction on~$\nu$ that~$L_\nu$ is a zigzag language.
The base case $\nu=0$ is trivial.
For $\nu\geq 1$, from the definition of the mapping~$\sigma$ and Lemma~\ref{lem:delete} we have $\{p(\pi)\mid \pi\in L_\nu\}=\{p(\sigma(F))\mid F\in\cE(G^{[\nu]})\}=\{\sigma(p(F))\mid F\in\cE(G^{[\nu]})\}=\{\sigma(F')\mid F'\in\cE(G^{[\nu-1]})\}=L_{\nu-1}$, and we know by induction that~$L_{\nu-1}$ is a zigzag language.

We first consider the case that $\nu$ is not an isolated vertex in~$G^{[\nu]}$, i.e., it is contained in a component with its neighbors.
Let $F'$ be the elimination forest for~$G^{[\nu]}-\nu=G^{[\nu-1]}$, let~$T$ be the elimination tree in~$F'$ containing these neighbors, let $x$ be the neighbor of~$\nu$ in~$G^{[\nu]}$ that is deepest in~$T$, and let $\pi\in S_{\nu-1}$ be such that $\sigma(F')=\pi$.
Then making $\nu$ the new root of~$T$ yields an elimination forest~$F$ for~$G^{[\nu]}$ with $\sigma(F)=c_1(\pi)$ by part~(a) in the definition of~$\sigma$, and making $\nu$ a leaf of~$x$ in~$T$ yields an elimination forest~$F$ with $\sigma(F)=c_\nu(\pi)$ by part~(c) in the definition of~$\sigma$.
It follows that $c_1(\pi)\in L_\nu$ and $c_\nu(\pi)\in L_\nu$ for all $\pi\in L_{\nu-1}$, so $L_\nu$ is a zigzag language by condition~(z1).

It remains to consider the case that $\nu$ is an isolated vertex in~$G^{[\nu]}$.
In this case $\nu$ is an isolated vertex in every elimination forest~$F$ for~$G^{[\nu]}$, and consequently $L_\nu=\{c_\nu(\pi)\mid \pi\in L_{\nu-1}\}$ by part~(c) in the definition of~$\sigma$, so $L_\nu$ is a zigzag language by condition~(z2).

This completes the proof of the lemma.
\end{proof}

By Theorem~\ref{thm:jump} and Lemma~\ref{lem:chordal-zigzag}, we can thus apply Algorithm~J to generate all permutations from~$\Pi(G)$.
We now show that under the bijection~$\sigma:\cE(G)\rightarrow \Pi(G)$, the preimages of minimal jumps in permutations performed by Algorithm~J correspond to tree rotations in elimination forests for~$G$.

The following result was proved in~\cite{MR4598046}.
We say that a jump of a value~$j$ in a permutation~$\pi\in S_n$ is \emph{clean}, if for every $k=j+1,\ldots,n$, the value~$k$ is either to the left or right of all values smaller than~$k$ in~$\pi$.

\begin{lemma}[{\cite[Lemma~24~(d)]{MR4598046}}]
\label{lem:clean-jumps}
For any zigzag language $L_n\seq S_n$, all jumps performed by Algorithm~J are clean.
\end{lemma}

We say that a rotation of an edge~$\{i,j\}$, $i<j$, is \emph{clean} if for all $k=j+1,\ldots,n$, the vertices smaller than~$k$ in its elimination tree are either all descendants of~$k$, or none of them is a descendant of~$k$.

\begin{lemma}
\label{lem:jump-rot}
Let $G=([n],E)$ be a PEO graph.
Clean minimal jumps of values~$j\in[n]$ in~$\Pi(G)$ are in one-to-one correspondence with clean rotations of edges~$\{i,j\}$, $i<j$, in~$\cE(G)$.
Moreover, every minimal jump of a value~$j\in[n]$ in~$\Pi(G)$ corresponds to the rotation of an edge~$\{i,j\}$, $i<j$, in~$\cE(G)$.
\end{lemma}

Complementing Lemma~\ref{lem:jump-rot}, Figure~\ref{fig:not-jump} shows a tree rotation in~$\cE(G)$ that is \emph{not} a jump in~$\Pi(G)$.
The vertex~$n=4$ is not a root or leaf of the elimination trees~$T$ or~$T'$, and while $\sigma(p(T))=123$ and~$\sigma(p(T'))=312$ differ in a minimal jump of the value~3, inserting~4 into these permutations gives permutations~$\sigma(T)=1423$ and~$\sigma(T')=3142$ that differ in a cyclic right rotation, but as $4>3$ this is not a jump.

\begin{figure}
\centering
\begin{tikzpicture}[scale=.8,auto,swap]
\begin{scope}[yshift=1cm]
\node[vertex] (4) at (0,0) {4};
\node[vertex] (2) at (1,0) {2};
\node[vertex, fill=red!20] (1) at (2,0) {1};
\node[vertex, fill=red!20] (3) at (3,0) {3};
\foreach \source/\dest in {4/2,2/1,1/3} \path[edge] (\source) -- (\dest);
\node at (-1,0) {$G$};
\end{scope}
\begin{scope}[xshift=6cm]
\node[vertex, fill=red!20] (1) at (0,2) {1};
\node[vertex] (4) at (-1,1) {4};
\node[vertex, fill=red!20] (3) at (1,1) {3};
\node[vertex] (2) at (-1,0) {2};
\foreach \source/\dest in {1/4,1/3,4/2} \path[edge] (\source) -- (\dest);
\node at (-1,2) {$T$};
\node at (0,-1) {$\sigma(T)=1423$};
\end{scope}
\begin{scope}[xshift=10cm]
\node[vertex, fill=red!20] (3) at (0,2) {3};
\node[vertex, fill=red!20] (1) at (0,1) {1};
\node[vertex] (4) at (0,0) {4};
\node[vertex] (2) at (0,-1) {2};
\foreach \source/\dest in {3/1,1/4,4/2} \path[edge] (\source) -- (\dest);
\node at (-1,2) {$T'$};
\node at (0,-2) {$\sigma(T')=3142$};
\end{scope}
\end{tikzpicture}
\iffalse
\begin{scope}[yshift=-4cm]
\draw[oriented edge] (3,4.3) -- (3,3.2);
\foreach \pos/\name in {{(1,2)/a}, {(2,2)/b}, {(3,1)/c}, {(4,1)/d}, {(5,1)/e}, {(6,1)/f}} \node[vertex] (\name) at \pos {$\name$};
\node[vertex, fill=red!20] (p) at (2,3) {$j$};
\node[vertex, fill=red!20] (v) at (4,2) {$i$};
\foreach \source/\dest in {p/a,p/b,p/v,v/c,v/d,v/e,v/f} \path[edge] (\source) -- (\dest);
\node at (6,2.5) {$T'$};
\end{scope}

(a) & (b)
\end{tabular}
\fi
\caption{The rotation of the edge~$\{3,1\}$ is not a jump in the corresponding permutations.
}
\label{fig:not-jump}
\end{figure}
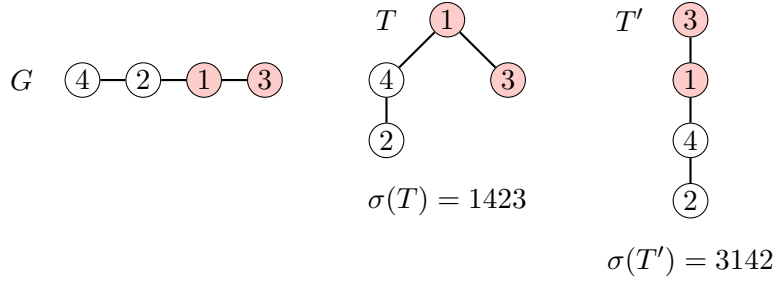

\begin{figure}
\centering
\includegraphics[page=2,scale=0.9]{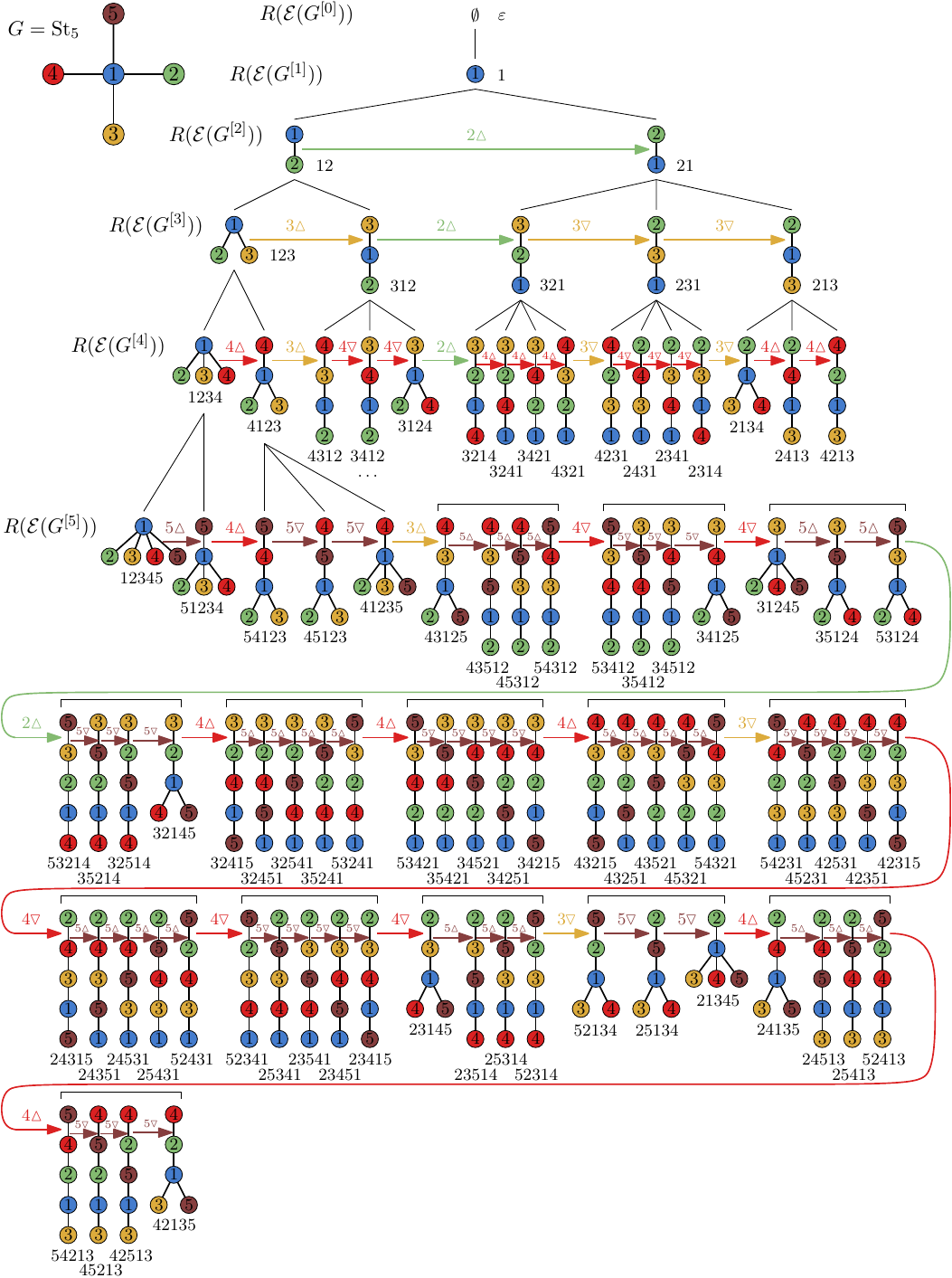}
\caption{The tree~$\fT(G)$ of all elimination trees for~$G^{[\nu]}$ for $\nu=0,\ldots,n$ for $G=\Star_5$, with the elimination trees in each depth listed in the order produced by Algorithm~R.
The bottom level nodes are split across four lines (brackets mark siblings), and for them not all parent-child relations are shown.}
\label{fig:Jtree}
\end{figure}

\begin{proof}
We argue by induction on~$\nu=0,\ldots,n$ that clean minimal jumps of values~$j\in[\nu]$ in~$\Pi(G^{[\nu]})$ are in one-to-one correspondence with clean rotations of edges~$\{i,j\}$, $i<j$, in~$\cE(G^{[\nu]})$, and every minimal jump of a value~$j\in[\nu]$ in~$\Pi(G^{[\nu]})$ corresponds to the rotation of an edge~$\{i,j\}$, $i<j$, in~$\cE(G)$.
The base case $\nu=0$ is trivial.

For the induction step we first consider two elimination forests~$F$ and~$F'$ for~$G^{[\nu]}$ such that~$\pi:=\sigma(F)\in S_\nu$ and~$\pi':=\sigma(F')\in S_\nu$ differ in a minimal jump of some value~$j\in[\nu]$.

{\bf Case (i):} We have $j=\nu$, so the jump is trivially a clean jump.
It follows that $p(\pi)=p(\pi')$ and therefore $p(F)=p(F')=:P\in\cE(G^{[\nu-1]})$.
Let $(x_1,\ldots,x_\lambda)$ be the corresponding insertion path in a tree of the elimination forest~$P$.

By the definition of~$\sigma$, the permutations $\sigma(c_i(P))$ for $i=1,\ldots,\lambda+1$ can be linearly ordered, where $\nu$ appears before $x_1$ if $i=1$, between $x_{i-1}$ and~$x_i$ if $2\leq i\leq \lambda$, and after $x_\lambda$ if $i=\lambda+1$.
As $\pi$ and~$\pi'$ differ in a minimal jump of~$\nu$, we have $\pi=\sigma(c_i(P))$ and~$\pi'=\sigma(c_{i+1}(P))$ for some $1\leq i\leq \lambda+1$.
By Lemma~\ref{lem:insert}, the elimination forests $F=c_i(P)$ and $F'=c_{i+1}(P)$ differ in a rotation of the edge~$\{x_i,\nu\}$, which is trivially a clean rotation.

{\bf Case (ii):} We have $j<\nu$.
By the definition of jumps, the value~$n$ is either to the left or to the right of~$j$ in both~$\pi$ and~$\pi'$.
Consequently, $p(\pi)$ and $p(\pi')$ differ in a minimal jump of the same value~$j$.
By induction, the elimination forests~$p(F)$ and~$p(F')$ differ in the rotation of an edge~$\{i,j\}$, $i<j$.
As $\pi$ and~$\pi'$ are linear extensions of~$F$ and~$F'$, the values~$j$ and~$i$ have different relative positions in~$\pi$ and~$\pi'$, i.e., the value~$j$ jumps over the value~$i$.
Consequently, the value~$n$ is either to the left or to the right of~$j$ and~$i$ in both~$\pi$ and~$\pi'$.
From this we obtain that~$F$ and~$F'$ differ in a rotation of the edge~$\{i,j\}$.

Now suppose that the minimal jump of the value~$j$ in the permutation~$\pi$ is clean, meaning that for every $k=j+1,\ldots,\nu$, the value~$k$ is either to the left or right of all values smaller than~$k$ in~$\pi$.
We claim that every $k=j+1,\ldots,\nu$ is either the root or a leaf of the elimination tree~$P:=p^{\nu-k}(F)$ containing~$j$.
Indeed, if $k$ is left of all values smaller than~$k$, then $k$ must be the root of~$P$, as in cases~(b) and~(c) of the definition of~$\sigma$, the value~$k$ is not at the first position of~$\sigma(P)$.
Similarly, if $k$ is right of all values smaller than~$k$, then $k$ must be a leaf of~$P$, as in cases~(a) and~(b) of the definition of~$\sigma$, the value~$k$ is not at the last position of~$\sigma(P)$.
It follows that in~$F$, the vertices smaller than~$k$ in its elimination tree are either all descendants of~$k$, or none of them is a descendant of~$k$, i.e., the rotation of the edge~$\{i,j\}$ in~$F$ is clean.

To complete the induction step, we now consider two elimination forests~$F$ and~$F'$ for~$G^{[\nu]}$ that differ in the clean rotation of an edge~$\{i,j\}$, $i<j$.
If $j=\nu$, then $p(F)=p(F')=:P\in\cE(G^{[\nu-1]})$, i.e., the permutations $\pi:=\sigma(F)$ and $\pi':=\sigma(F')$ are obtained from~$\sigma(P)$ by inserting the value~$\nu$ before and after one of the values~$x_k$, $1\leq k\leq \lambda$, where $(x_1,\ldots,x_\lambda)$ is the insertion path of~$P$, i.e., the permutations~$\pi$ and~$\pi'$ differ in a minimal jump of the value~$\nu$.

If $j<\nu$, then as the rotation is clean, the vertex~$\nu$ is the root or a leaf of its elimination tree.
Consequently, the elimination forests $p(F)$ and~$p(F')$ of~$G^{[\nu-1]}$ differ in a clean rotation of the edge~$\{i,j\}$, and by induction the permutations~$\pi:=\sigma(p(F))$ and~$\pi':=\sigma(p(F'))$ differ in a clean jump of the value~$j$.
As $\nu$ is the root or a leaf of its elimination tree, $\sigma(F)$ and~$\sigma(F')$ are obtained from~$\pi$ and~$\pi'$, respectively, by either prepending or appending~$\nu$; recall cases~(a) and~(c) of the definition of~$\sigma$.
We conclude that~$\sigma(F)$ and~$\sigma(F')$ differ in a clean minimal jump of the value~$j$.

This completes the induction step and the proof of the lemma.
\end{proof}

By Lemma~\ref{lem:jump-rot}, all minimal jumps in permutations from~$\Pi(G)$ correspond to tree rotations in elimination forests from~$\cE(G)$.
By translating Algorithm~J to operate directly on elimination forests, we thus obtain Algorithm~R presented in Section~\ref{sec:results-algo}.
In fact, we do not need to restrict Algorithm~R to clean rotations.
This is justified by the following analogue of Lemma~\ref{lem:clean-jumps}, which follows from Lemma~\ref{lem:clean-rots-j} stated and proved in Section~\ref{sec:necessity} below.

\begin{lemma}
\label{lem:clean-rots}
For any PEO graph~$G=([n],E)$, all up- or down-rotations performed by Algorithm~R are clean.
\end{lemma}

With these lemmas at hand, we are now in position to prove Theorem~\ref{thm:jump-elim}.

\begin{proof}[Proof of Theorem~\ref{thm:jump-elim}]
Let $G=([n],E)$ be a PEO graph, and consider the mapping $\sigma:\cE(G)\rightarrow S_n$ defined at the beginning of this section and the corresponding set of permutations $\Pi(G)=\{\sigma(F)\mid F\in\cE(G)\}$.
The set $\Pi(G)\seq S_n$ is a zigzag language by Lemma~\ref{lem:chordal-zigzag}, so by Theorem~\ref{thm:jump} we can apply Algorithm~J to generate all permutations from~$\Pi(G)$.
By Lemma~\ref{lem:clean-jumps}, all jumps performed by Algorithm~J in permutations from~$\Pi(G)$ are clean.
By Lemma~\ref{lem:jump-rot}, clean jumps in~$\Pi(G)$ are in one-to-one correspondence with clean rotations in elimination forests from~$\cE(G)$.
We can thus translate Algorithm~J to operate on elimination forests, which yields Algorithm~R presented in Section~\ref{sec:results-algo}.
This algorithm as stated in principle allows rotations that are not clean, which however, never occur by Lemma~\ref{lem:clean-rots}.
Translating Theorem~\ref{thm:jump} to elimination forests thus proves Theorem~\ref{thm:jump-elim}.
The initial elimination forest~$F_0$ that is obtained by removing vertices in increasing order is $F_0=\sigma^{-1}(\ide_n)$, i.e., it corresponds to the identity permutation as the elimination order.
\end{proof}

An example output of Algorithm~R in the case where $G$ is a tree on five vertices is shown in Figure~\ref{fig:Jtree}.
We write $R(\cE(G))$ for the ordering of elimination forests generated by Algorithm~R with initial elimination forest $F_0=\sigma^{-1}(\ide_n)$.
We clearly have $R(\cE(G))=\sigma^{-1}(J(\Pi(G)))$.

By~\eqref{eq:JLn12}, the ordering~$R(\cE(G))$ of elimination forests produced by the algorithm can be described recursively as follows; see Figure~\ref{fig:Jtree}:
From the sequence~$R(\cE(G^{[n-1]}))$ of all elimination forests for~$G^{[n-1]}$ we obtain the sequence~$R(\cE(G^{[n]}))$ either by replacing each elimination forest~$F\in\cE(G^{[n-1]})$ by the sequence~$c_i(F)$ for $i=\lambda(F)+1,\lambda(F),\ldots,1$ or $i=1,\ldots,\lambda(F)+1$, alternatingly in decreasing or increasing order (in the case~\eqref{eq:JLn1}), or by adding~$n$ as an isolated vertex to each elimination forest (in the case~\eqref{eq:JLn2}). 
In this way, the sequences~$R(\cE(G^{[\nu]}))$ for $\nu=0,\ldots,n$ turn the unordered tree~$\fT(G)$ into an ordered tree, and the ordering of elimination forests on each level of the tree corresponds to a Hamilton path on the graph associahedron~$\cA(G^{[\nu]})$.

\section{Efficient implementation and proof of Theorem~\ref{thm:algo}}
\label{sec:algo}

In this section we prove Theorem~\ref{thm:algo} by presenting an efficient implementation of Algorithm~R, following the ideas sketched in Section~\ref{sec:results-impl}.
In particular, this implementation is history-free, i.e., it does not require any bookkeeping of previously visited elimination forests.

\subsection{An efficient algorithm for chordal graphs}

For a PEO graph~$G=([n],E)$, we say that a vertex~$j\in[n]$ is \emph{rotatable}, if $j$ is not an isolated vertex in~$G^{[j]}$; see Figure~\ref{fig:alpha}.

\begin{wrapfigure}{r}{0.55\textwidth}
\centering
\includegraphics{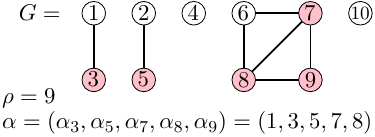}
\caption{Definition of rotatable vertices (highlighted) and of~$\alpha$ and~$\rho$.}
\label{fig:alpha}
\end{wrapfigure}
Clearly, $j$ is rotatable if and only if it is not the smallest vertex in its connected component.
For every rotatable vertex~$j\in[n]$, we let $\alpha_j$ denote the next smaller rotatable vertex or $\alpha_j=1$ if there is no such vertex.
Moreover, we let $\rho$ denote the largest rotatable vertex.
If $G$ is connected all vertices $j=2,\ldots,n$ are rotatable and we have $\rho=n$ and $\alpha_j=j-1$.

\begin{algo}{\bfseries Algorithm~H}{History-free rotations}
This algorithm generates all elimination forests for a PEO graph~$G=([n],E)$ by rotations in the same order as Algorithm~R.
It maintains the current elimination forest in the variable~$F$, the insertion path~$Q_j$ for the vertex~$j\in[n]$ in $p^{n-j+1}(F)$ and an index~$q_j$ to a vertex on this path, as well as auxiliary arrays $o=(o_1,\ldots,o_n)$ and $s=(s_1,\ldots,s_n)$.
\begin{enumerate}[label={\bfseries H\arabic*.}, leftmargin=8mm, noitemsep, topsep=3pt plus 3pt]
\item{} [Initialize] Set $F\gets \sigma^{-1}(\ide_n)$, and $o_j\gets \diru$ and $s_j\gets j$ for $j=1,\ldots,n$.
\item{} [Visit] Visit the current elimination forest~$F$.
\item{} [Select vertex] Set $j\gets s_\rho$, and terminate if $j=1$.
\item{} [Insertion path] If $o_j=\dird$ and $q_j=0$, compute the insertion path~$Q_j$ for the vertex~$j$ in its elimination tree in~$F$. 
\item{} [Rotate vertex] In the current elimination forest~$F$, if $o_j=\diru$ perform an up-rotation of~$j$, whereas if $o_j=\dird$ perform a down-rotation of~$j$ and set $q_j\gets q_j+1$.
\item{} [Update $o$ and $s$] Set $s_\rho\gets \rho$.
If $o_j=\diru$ and $j$ is the root of a tree in~$F$ or its parent is larger than~$j$ set $o_j\gets\dird$ and $q_j\gets 0$, or if $o_j=\dird$ and $j$ is a leaf or its children are all larger than~$j$ set $o_j\gets \diru$, and in both cases set $k\gets \alpha_j$, $s_j\gets s_k$ and $s_k\gets k$. Go back to~H2.
\end{enumerate}
\end{algo}

Algorithm~H is a straightforward translation of the history-free algorithm for zigzag languages presented in~\cite{MR4598046} to elimination forests.
The array $s=(s_1,\ldots,s_n)$ simulates a stack that selects the vertex to rotate up and down in the next step.
In fact, only array entries~$s_j$ for which $j$ is rotatable are modified and used in the algorithm.
Specifically, the vertex~$j$ rotated in the current step is extracted from the top of the stack by the instruction $j\gets s_\rho$ (line~H3).
The direction in which each vertex is rotating next (up or down) is determined by the array $o=(o_1,\ldots,o_n)$ (line~H5), where again only entries~$o_j$ for which $j$ is rotatable are relevant.
The arrays~$o$ and~$s$ are updated in line~H6.
Specifically, after rotating~$j$ up, the direction of rotation for~$j$ is reversed if $j$ has become root of its elimination tree or its parent is larger than~$j$.
Similarly, after rotating~$j$ down, the direction of rotation is reversed if $j$ has become a leaf of its elimination tree or its children are all larger than~$j$.

When a vertex~$j$ starts a sequence of down-rotations, we first compute its insertion path~$Q_j$ (line~H4); recall the definitions from Section~\ref{sec:del-ins} and Figure~\ref{fig:insert}.
While~$Q_j$ is defined in~$p^{n-j+1}(F)$, we compute it directly in~$F$, by ignoring any of the vertices $j+1,\ldots,n$.
We maintain the current position of~$j$ on~$Q_j$ in the variable~$q_j$, by incrementing $q_j$ after every down-rotation of~$j$ in line~H5.
During up-rotations, this information is irrelevant and not maintained.
The counter~$q_j$ is reset to~0 in line~H6 after the vertex~$j$ completes a sequence of up-rotations, which triggers recomputation of~$Q_j$ in line~H4, once $j$ is again selected for a down-rotation.
The insertion path~$Q_j$ can be computed efficiently as follows:
We iterate over each neighbor~$i$ of~$j$ in~$G^{[j]}$, and we traverse the path from~$i$ in~$F$ upwards until we reach~$j$ or a previously encountered neighbor of~$j$ in~$G^{[j]}$, and by reversing and concatenating the resulting subpaths to~$Q_j$.
Following this approach, the computation of~$Q_j$ takes time~$\cO(\lambda)$, where $\lambda=\lambda(F)$ is the number of vertices of~$Q_j$.
This is amortized to~$\cO(1)$ by the $\lambda$ subsequent down-rotation steps involving the vertex~$j$, which do not incur any further computations in step~H4.

The operations in line~H3 and~H6 can be implemented straightforwardly in constant time.
It remains to discuss how to implement the tree rotation operation in line~H5 efficiently; recall the definitions from Section~\ref{sec:rot} and Figure~\ref{fig:rotation}.
For the rotation of an edge~$\{i,j\}$, the key problem is to decide which of the subtrees of the lower of the two vertices~$i$ and~$j$ change their parent.
During an up-rotation of~$j$, i.e., $i$ is the parent of~$j$, a subtree of~$j$ rooted at a vertex~$k$ changes its parent from~$j$ to~$i$ if and only if $k<j$, or $k>j$ and $\{k,i\}$ is an edge in~$G$.
Similarly, during a down-rotation of~$j$, i.e., $i$ is the unique child of~$j$ that is smaller than~$j$, a subtree of~$i$ rooted at a vertex~$k$ changes its parent from~$i$ to~$j$ if and only if $k<j$ and $k$ lies on the insertion path~$Q_j$ of~$j$ (specifically, at position~$p_j+1$), or $k>j$ and $\{k,j\}$ is an edge in~$G$.
With the help of the precomputed insertion path~$Q_j$, these steps can be performed in time~$\cO(\ell)$, where $\ell\leq \sigma(G)$ is the number of children of the lower of the two vertices~$i$ and~$j$.
For this we require constant-time adjacency queries in~$G$, which is achieved by storing an $n\times n$ adjacency matrix.

Combining these observations proves the runtime bound~$\cO(\sigma)$, $\sigma=\sigma(G)$, stated in the first part of Theorem~\ref{thm:algo} for chordal graphs~$G$.

Testing whether an arbitrary graph~$G$ is chordal, and if so computing a PEO for~$G$, can be done in time~$\cO(m+n)$ by lexicographic breadth-first-search~\cite{MR408312}, where $m$ is the number of edges of~$G$.
This is dominated by the time~$\cO(n^2)$ it takes to initialize the adjacency matrix of~$G$.

\subsection{A loopless algorithm for trees}

We now describe how to improve the algorithm to a loopless algorithm for trees~$G$.
The loopless algorithm is based on the observation that in a tree~$G=([n],E)$, for every rotation of an edge~$\{i,j\}$ in an elimination tree for~$G$, \emph{at most one} subtree changes its parent; see Figure~\ref{fig:rot-tree}.
Specifically, this subtree corresponds to the unique subtree between~$i$ and~$j$ in~$G$.
We can thus obtain a speed-up by introducing the following auxiliary data structures:
For each vertex~$j\in[n]$ in the current elimination tree we keep track of the unique child of~$j$ that is smaller than~$j$ (recall Lemma~\ref{lem:smaller-child}).
Upon a down-rotation of~$j$, this allows us to quickly determine the child~$i$ of~$j$ for which to perform the rotation of the edge~$\{i,j\}$ (this is not needed for up-rotations of~$j$).
Also, we precompute an $n\times n$ matrix~$\beta$ where $\beta_{i,j}$ is the unique neighbor of~$i$ in~$G$ in the direction of~$j$.
This matrix is only based on~$G$ and does not change throughout the algorithm.
Furthermore, we maintain an $n\times n$ matrix~$\gamma$ where $\gamma_{i,a}=j$ if in the current elimination tree the vertex~$i$ has a child~$j$, and in~$G$ the vertex~$j$ is reached from~$i$ in the direction of the neighbor~$a$ of~$i$, and $\gamma_{i,a}=0$ otherwise.

With these data structures, each tree rotation can be performed in constantly many steps as described by the function $\rotatej()$ shown below; see Figure~\ref{fig:rot-tree} for illustration.
The auxiliary function $\deli(i)$ deletes vertex~$i$ from its elimination tree, and $\insi(i,j)$ inserts $i$ as a child of~$j$.

\begin{figure}[t!]
\centering
\includegraphics{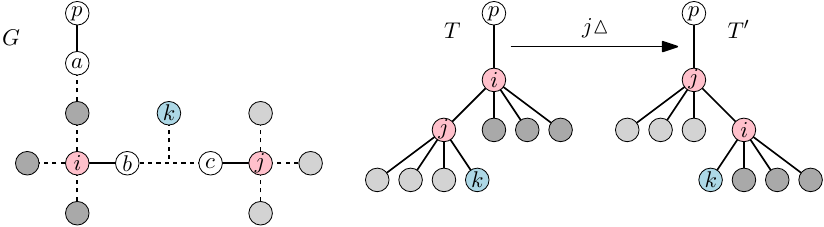}
\caption{Elimination tree rotation in constant time when $G$ is a tree.}
\label{fig:rot-tree}
\end{figure}

\begin{algo}{{\bfseries Function} $\rotatej(j{\diru}\!)$}{Up-rotate $j$}
\begin{enumerate}[label={\bfseries \arabic*.}, leftmargin=6mm, noitemsep, topsep=3pt plus 3pt]
\item{}[Prepare] Set $i$ to be the parent of~$j$, and set $p$ to be the parent of~$i$ or $p\gets 0$ if $i$ is root of its elimination tree.
Moreover, set $a\gets \beta_{p,i}$, $b\gets \beta_{i,j}$, $c\gets \beta_{j,i}$ and $k\gets \gamma_{j,c}$.
\item{}[Rotate]
Call $\deli(j)$.
If $k\neq 0$ call $\deli(k)$ and $\insi(k,i)$.
If $p\neq 0$ call $\deli(i)$ and $\insi(j,p)$, otherwise set $j$ to be the new root.
Call $\insi(i,j)$.
\item{}[Update] Set $\gamma_{p,a}\gets j$, $\gamma_{i,b}\gets k$ and $\gamma_{j,c}\gets i$.
\end{enumerate}
\end{algo}

This completes the proof of Theorem~\ref{thm:algo}.

We implemented the algorithm for chordal graphs and the optimized version for trees on the Combinatorial Object Server~\cite{cos_elim}, and we encourage the reader to experiment with this code.

\section{Hamilton cycles and proof of Theorem~\ref{thm:cycle-summary}}
\label{sec:cycle}

Recall from Theorem~\ref{thm:cycle} that the graph associahedron~$\cA(G)$ has a Hamilton cycle for any graph~$G$ with at least two edges.
By Theorem~\ref{thm:jump-elim}, Algorithm~R produces a Hamilton path for chordal graphs~$G$.
We now derive several conditions on~$G$ under which the end vertices of the Hamilton path produced by the algorithm are adjacent in~$\cA(G)$, i.e., the algorithm produces a Hamilton cycle.
In particular, combining Theorems~\ref{thm:2conn} and~\ref{thm:tree} stated below immediately yields Theorem~\ref{thm:cycle-summary} stated in Section~\ref{sec:results-cycle}.

We start with a positive result.
For any integer~$k\geq 1$, a graph is \emph{$k$-connected}, if it has at least $k+1$~vertices and removing any set of $k-1$~vertices yields a connected graph.

\begin{theorem}
\label{thm:2conn}
If $G=([n],E)$ is a 2-connected PEO graph, then $R(\cE(G))$ is cyclic.
\end{theorem}

On the other hand, for trees~$G$, which are 1-connected, our algorithm does not produce a Hamilton cycle.

\begin{theorem}
\label{thm:tree}
If $G=([n],E)$, $n\geq 4$, is a PEO tree, then the ordering $R(\cE(G))$ is not cyclic.
\end{theorem}

\begin{wrapfigure}{r}{0.4\textwidth}
\centering
\includegraphics{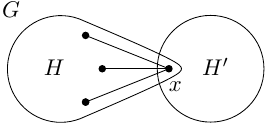}
\caption{Definition of $H$-appendix.}
\vspace{-7mm}
\label{fig:appendix}
\end{wrapfigure}
Lastly, we identify several classes of graphs that are not 2-connected and not trees, and for which we still obtain a Hamilton cycle.

For any graph~$H$ that is 2-connected and that has a vertex~$x$ such that $H-x$ is 2-connected or a single edge, we say that a graph~$G$ has an \emph{$H$-appendix}, if $G$ is obtained from an arbitrary graph~$H'$ with at least two vertices by gluing $H$ with the vertex~$x$ onto one of the vertices of~$H'$; see Figure~\ref{fig:appendix}.
The smallest suitable such~$H$ would be a triangle.
Note that $H'$ is not required to be connected in this definition.

We now also define the notion of a simplicial vertex for graphs~$G$ that have no ordering of their vertices, slightly overloading notation.
Specifically, a vertex in a graph~$G$ is called \emph{simplicial}, if it forms a clique together with all its neighbors in the graph.

\begin{figure}
\centering
\includegraphics{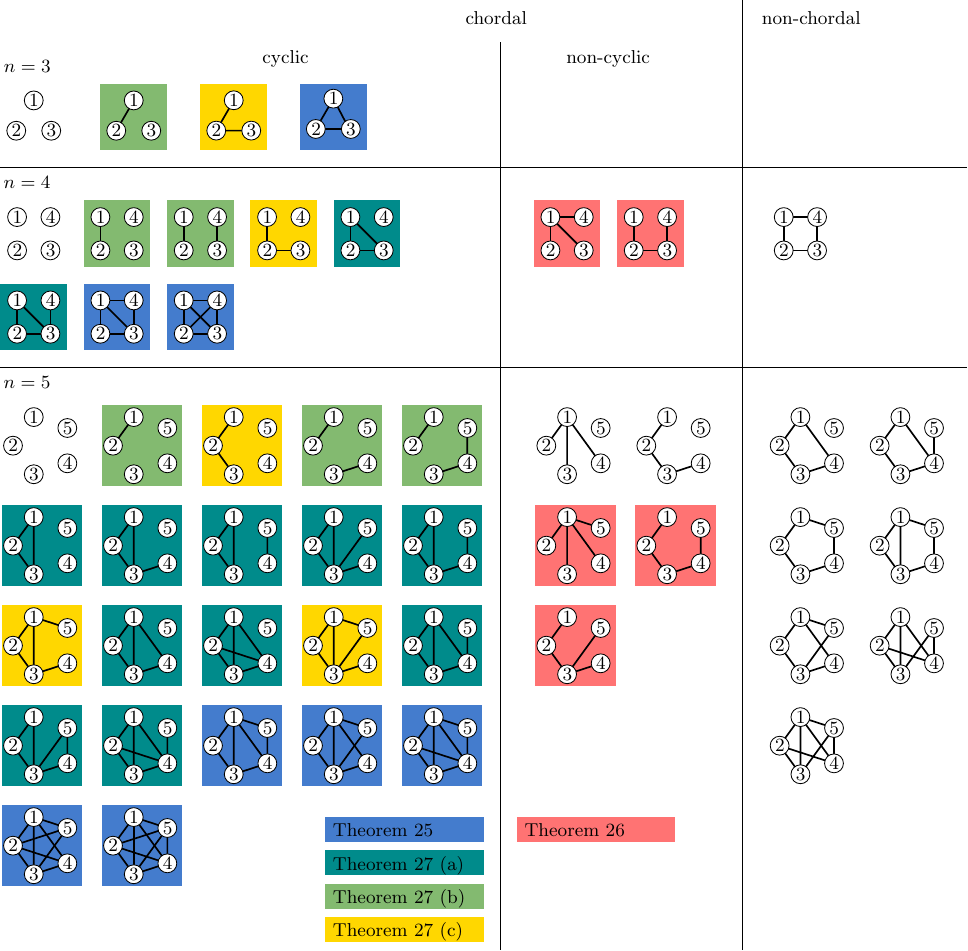}
\caption{Applying Theorems~\ref{thm:2conn}--\ref{thm:peo} to all non-isomorphic graphs on $n=3,4,5$ vertices; cf.~Figure~\ref{fig:asso4}.}
\label{fig:cyclic}
\end{figure}

\todo{Make sure theorem numbers in the figure are up to date.}

\begin{theorem}
\label{thm:peo}
The following chordal graphs~$G$ admit a PEO such that $R(\cE(G))$ is cyclic:
\begin{enumerate}[label=(\alph*),leftmargin=6mm, noitemsep, topsep=3pt plus 3pt]
\item $G$ has an $H$-appendix;
\item $G$ has an isolated edge;
\item $G$ has a simplicial vertex~$y$, and $G-y$ is as in~(a), (b) or a 2-connected chordal graph.
\end{enumerate}
\end{theorem}

Note that Theorems~\ref{thm:2conn} and~\ref{thm:tree} hold for arbitrary PEOs for 2-connected chordal graphs and trees, respectively, whereas Theorem~\ref{thm:peo} talks about one suitably chosen PEO for the graph.
Figure~\ref{fig:cyclic} shows the result of applying Theorems~\ref{thm:2conn}--\ref{thm:peo} to all graphs on $n=3,4,5$ vertices.

\subsection{Counting elimination trees}

For proving our results, we need the following lemma.

\begin{lemma}[{\cite[Lemma~4]{MR4391718}}]
\label{lem:cyclic}
For any zigzag language of permutations~$L_n\seq S_n$, in the ordering of permutations~$J(L_n)$ generated by Algorithm~J, the first and last permutation are related by a minimal jump if and only if $|L_\nu|=|\{p^{n-\nu}(\pi)\mid \pi\in L_n\}|$ is even for all~$\nu=2,\ldots,n-1$.
\end{lemma}

Motivated by this lemma, for any graph~$G$ we write $e(G):=|\cE(G)|$ for the number of elimination forests for~$G$.
This function generalizes all the counting functions corresponding to the combinatorial classes mentioned in Section~\ref{sec:objects}, i.e., Catalan numbers, factorial numbers etc.
The number~$e(G)$ can be computed recursively as follows:
If $G$ is connected we have
\begin{subequations}
\label{eq:eG}
\begin{equation}
\label{eq:eG-sum}
e(G)=\sum_{x\in G} e(G-x),
\end{equation}
and if $G$ is disconnected we have
\begin{equation}
\label{eq:eG-prod}
e(G)=\prod_{H\in\cC(G)} e(H),
\end{equation}
\end{subequations}
where $\cC(G)$ is the set of connected components of~$G$.

\begin{lemma}
\label{lem:even}
Any connected graph with an even number of vertices has an even number of elimination trees.
\end{lemma}

\begin{proof}
Consider a connected graph~$G$ with an even number~$n$ of vertices.
As $G$ is connected, all of its elimination trees have $n-1$ edges, so there are $n-1$ possible tree rotations and consequently the graph associahedron~$\cA(G)$ is an $(n-1)$-regular graph by Lemma~\ref{lem:Gasso-edges}.
By the handshaking lemma, $e(G)\cdot(n-1)$ must be even, and as $n-1$ is odd, we obtain that $e(G)$ is even.
\end{proof}

\begin{lemma}
\label{lem:2-conn}
Any 2-connected graph has an even number of elimination trees.
\end{lemma}

\begin{proof}
Consider a 2-connected graph~$G$ with $n\geq 3$ vertices.
If $n$ is even, then the claim follows directly from Lemma~\ref{lem:even}.
If $n$ is odd, then removing any vertex~$x$ of~$G$ leaves a connected graph~$G-x$ with an even number~$n-1$ of vertices.
Consequently, by Lemma~\ref{lem:even} and by~\eqref{eq:eG-sum}, $e(G)$ is a sum of even numbers, so it is even.
\end{proof}

We will use the following well-known result due to Dirac about chordal graphs.

\begin{lemma}[\cite{MR130190}]
\label{lem:simplicial}
Any chordal graph with at least two vertices has at least two simplicial vertices.
\end{lemma}

The next lemma asserts that in a chordal graph, we can choose any vertex to be first in a PEO.

\begin{lemma}
\label{lem:peo-first}
For any chordal graph~$G$ and any vertex~$x$ of~$G$, there is a PEO starting with~$x$.
\end{lemma}

\begin{proof}
We argue by induction on the number of vertices of~$G$.
By Lemma~\ref{lem:simplicial}, $G$ has two simplicial vertices, at least one of which is different from~$x$, and we pick one such vertex~$y$.
By induction, $G-y$ has a PEO starting with~$x$.
As $y$ is simplicial in~$G$, we can append~$y$ to that ordering and obtain a PEO of~$G$ that starts with~$x$.
\end{proof}

\begin{lemma}
\label{lem:peo-triangle}
If $G$ is chordal and has an $H$-appendix, then it has a PEO that starts with the vertices of~$H$.
\end{lemma}

\begin{proof}
Let $x$ be the vertex of~$G$ that joins~$H$ to the rest of the graph, and let~$H'$ denote the subgraph of~$G$ induced by~$x$ and the vertices not in~$H$; see Figure~\ref{fig:appendix}.
Clearly~$H$ is chordal, so it has a PEO by Lemma~\ref{lem:chordal-PEO}.
Moreover, $H'$ has a PEO starting with~$x$ by Lemma~\ref{lem:peo-first}.
The concatenation of the first PEO and the second one with $x$ removed is the desired PEO for~$G$.
\end{proof}

It is well known that in a connected graph, any prefix of a PEO is also connected.
Interestingly, this property of PEOs in chordal graphs generalizes to $k$-connectedness.

\begin{lemma}
\label{lem:peo-kconn}
If $G=([n],E)$ is a $k$-connected PEO graph, then $G^{[\nu]}$ is $\min\{\nu-1,k\}$-connected for all $\nu=2,\ldots,n$.
\end{lemma}

\begin{proof}
It suffices to prove the cases~$\nu>k$.
Suppose the lemma was false, and consider the smallest counterexample, i.e., the smallest possible~$k$, a $k$-connected PEO graph~$G=([n],E)$ and the smallest $\nu$ such that $G^{[\nu-1]}$ is $k$-connected and $G^{[\nu]}$ is not $k$-connected, which implies that $G^{[\nu]}$ is $(k-1)$-connected.
This means that the vertex~$\nu$ has a set~$N$ of precisely $|N|=k-1$ neighbors among the vertices~$1,\ldots,\nu-1$, because if it had more than that, then removing any $k-1$ of them from~$G^{[\nu]}$ would leave the graph connected, i.e., $G^{[\nu]}$ would be $k$-connected, and if it had less, then~$G^{[\nu]}$ would not be $(k-1)$-connected.
Let $x$ be one of the vertices from~$\{1,\ldots,\nu-1\}\setminus N$, which exists as $\nu-1>k-1$.
As $G$ is $k$-connected, there are $k$ vertex-disjoint paths in~$G$ between~$\nu$ and~$x$, at most $k-1$ of which use a vertex from~$N$, and we let $P$ be such a path starting at~$\nu$, ending at~$x$ and avoiding~$N$.
As $P$ avoids~$N$, all neighbors of~$\nu$ on~$P$ have a higher label than~$\nu$.
On the other hand, the last vertex~$x$ of~$P$ has a smaller label than~$\nu$.
It follows that the subgraph of~$G$ induced by the vertex set of~$P$ contains a shortest path that starts at~$\nu$, ends at~$x$ and contains a triple of vertices~$i,j,k$ in that order with $j>i,k$, which is a contradiction by Lemma~\ref{lem:no-peak}.
\end{proof}

\subsection{Proofs of Theorems~\ref{thm:2conn}--\ref{thm:peo}}

With the previous lemmas in hand, we are now ready to present the proofs of the theorems stated at the beginning of this section.

\begin{proof}[Proof of Theorem~\ref{thm:2conn}]
As $G$ is 2-connected, $G^{[2]}$ is 1-connected, i.e., a single edge, and $G^{[\nu]}$ is 2-connected for all $\nu=3,\ldots,n$ by Lemma~\ref{lem:peo-kconn}.
Consequently, by Lemmas~\ref{lem:even} and~\ref{lem:2-conn}, $|L_\nu|=|\sigma(\cE(G^{[\nu]}))|=e(G^{[\nu]})$ is even for all $\nu=2,\ldots,n$.
The theorem now follows by combining Lemmas~\ref{lem:jump-rot} and~\ref{lem:cyclic}.
\end{proof}

\begin{proof}[Proof of Theorem~\ref{thm:tree}]
As $G$ is a tree, the first three vertices in a PEO of~$G$ must be adjacent by Lemma~\ref{lem:no-peak}.
Consequently, the subgraph~$G^{[3]}$ is a path~$P$ on three vertices, which satisfies~$e(P)=5$, and between its two end vertices the path is either labeled $1,2,3$ or $2,1,3$ in a PEO.
The orderings~$R(\cE(G^{[3]}))$ for these two cases are shown in Figure~\ref{fig:P3-chordal}~(a) and~(b), respectively.
As $e(P)=5$ is odd, the first tree~$T$ in the ordering~$R(\cE(G^{[4]}))$ has the vertex~4 as a descendant of~1, whereas the last tree~$T'$ has 1 as a descendant of~4, but with~2 in between them.
It follows that $T$ and~$T'$ do not differ in a tree rotation, and the same is true for the first tree~$S$ and the last tree~$S'$ in~$R(\cE(G))$, as $p^{n-4}(S)=T$ and $p^{n-4}(S')=T'$.
\end{proof}

\begin{figure}
\centering
\includegraphics[page=3]{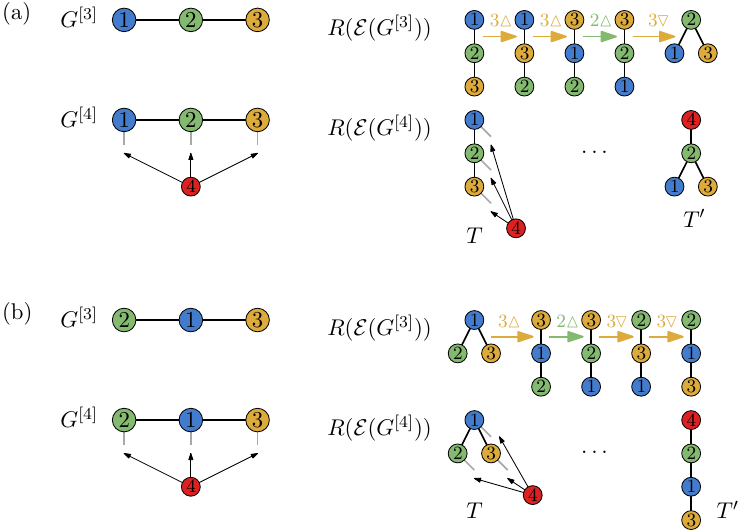}
\caption{The two rotation Gray codes for a path on three vertices.
See also Figure~\ref{fig:P3-non-chordal}.
}
\label{fig:P3-chordal}
\end{figure}

\begin{proof}[Proof of Theorem~\ref{thm:peo}]
We begin proving~(a).
Let $x$ and $H'$ be defined as in the proof of Lemma~\ref{lem:peo-triangle} (recall Figure~\ref{fig:appendix}), and consider a PEO of~$G$ that starts with the vertices in~$H$, which exists by this lemma.
We let $G=([n],E)$ denote the graph equipped with this PEO.

For every $\nu=3,\ldots,n$ that is contained in~$H$, $G^{[\nu]}$ is a subgraph of~$H$, which is assumed to be 2-connected.
Therefore, by Lemma~\ref{lem:peo-kconn}, $G^{[\nu]}$ is 2-connected, and therefore $e(G^{[\nu]})$ is even by Lemma~\ref{lem:2-conn}.
Moreover, $G^{[2]}$ is 1-connected, i.e., a single edge, by Lemma~\ref{lem:peo-kconn}, and therefore $e(G^{[2]})=2$ is even.

We now prove inductively that any induced subgraph~$H^+$ of~$G$ that includes~$H$ has an even number of elimination trees.
This claim implies in particular that $e(G^{[\nu]})$ is even for all $\nu=3,\ldots,n$ that are contained in~$H'-x$.

Note that it suffices to prove our claim for connected subgraphs~$H^+$, as for disconnected~$H^+$, the connected component containing~$H$ contributes an even factor to the product~\eqref{eq:eG-prod}.
So consider a connected induced subgraph~$H^+$ of~$G$ containing~$H$.
If the number of vertices of~$H^+$ is even, then the claim is immediate from Lemma~\ref{lem:even}.
If the number of vertices of~$H^+$ is odd, then \eqref{eq:eG-sum} gives
\begin{equation*}
e(H^+)=e(H^+-x)+\sum_{y\in H-x}e(H^+-y)+\sum_{y\in (H'-x)\cap H^+} e(H^+-y).
\end{equation*}
The graph $H^+-x$ is disconnected, but contains the 2-connected component $H-x$ (as required in the definition of $H$-appendix), and therefore $e(H^+-x)$ is even by Lemma~\ref{lem:2-conn} and~\eqref{eq:eG-prod}.
The graph $H^+-y$ for any $y\in H-x$ is connected, and therefore $e(H^+-y)$ is even by Lemma~\ref{lem:even}.
The graph $H^+-y$ for any $y\in (H'-x)\cap H^+$ is an induced subgraph of~$G$ containing~$H$, so $e(H^+-y)$ is even by induction.
We conclude that $e(H^+)$ is even.

Consequently, $L_\nu:=\Pi(G^{[\nu]})$ has even cardinality for all $\nu=2,\ldots,n$, so the theorem follows by combining Lemmas~\ref{lem:jump-rot} and~\ref{lem:cyclic}.

To prove~(b), we choose a PEO that labels the vertices of the isolated edge of~$G$ by~$1,2$, and the remaining vertices arbitrarily.
We have $|L_2|=2$ and $L_\nu$ has even cardinality for all $\nu=3,\ldots,n$ by~\eqref{eq:eG-prod}.

To prove~(c), we choose a PEO that labels the vertices of~$G-y$ by~$1,\ldots,n-1$ as explained in the proofs of~(a), (b) and the proof of Theorem~\ref{thm:2conn}, and we label the vertex~$y$ last.
We obtain that $L_\nu$ has even cardinality for all $\nu=2,\ldots,n-1$, so Lemma~\ref{lem:cyclic} still applies.
\end{proof}

\section{Chordality is necessary}
\label{sec:necessity}

By Theorem~\ref{thm:jump-elim}, Algorithm~R generates all elimination forests~$\cE(G)$ of a graph~$G$, provided that~$G$ is chordal with vertices in perfect elimination order.
We now show that the converse also holds, i.e., the algorithm fails to generate~$\cE(G)$ for non-chordal graphs~$G$ (Theorem~\ref{thm:necessity} below), regardless of the choice of the initial elimination forest~$F_0$.
As Algorithm~R is oblivious of the notion of chordal graphs, this is an interesting characterization of graph chordality.

The following lemma generalizes some of our earlier observations about the behavior of Algorithm~R to graphs that are not necessarily chordal.

\begin{lemma}
\label{lem:project}
Let $G=([n],E)$ be a graph with simplicial vertices~$s,s+1,\ldots,n$.
For any $F_0\in\cE(G)$, let $\cF=(F_0,\ldots,F_t)$ be the sequence of elimination forests produced by Algorithm~R through $t$ up- or down-rotations of vertices~$j\in\{s,\ldots,n\}$.
Then for every step~$F_{t-1}\rightarrow F_t$ in which Algorithm~R up- or down-rotates a vertex~$j\in\{s,\ldots,n\}$ we have the following:
\begin{enumerate}[label=(\alph*),leftmargin=6mm, noitemsep, topsep=3pt plus 3pt]
\item If $j<n$, then the vertex~$n$ is a root or a leaf in~$F_{t-1}$.
\item If $j=n$, then all further up-rotations or down-rotations of~$n$, respectively, until $n$ is a root or a leaf, yield elimination forests that are unvisited, i.e., that do not appear in the sequence~$\cF=(F_0,\ldots,F_t)$.
\end{enumerate}
\end{lemma}

By Lemma~\ref{lem:smaller-child-j}, in every elimination forest~$F$ of~$G$, every vertex~$j\in\{s,\ldots,n\}$ has at most one child that is smaller than~$j$.
Consequently, the deletion operation introduced in Section~\ref{sec:del-ins} can be applied to all vertices~$j=n,n-1,\ldots,s$ in~$F$, and the resulting forest~$p^{n-j+1}(F)$ is an elimination forest for the graph~$G^{[j-1]}$.

\begin{proof}
We argue by induction on~$t$.
The base case is trivial.

For the induction step let $t\geq 1$, and consider the elimination forests~$F_{t-1}$ and~$F_t$, and suppose that $F_t$ is obtained from~$F_{t-1}$ by up-rotating~$j\in\{s,\ldots,n\}$.
We need to argue that the invariants~(a) and~(b) are satisfied for~$F_{t+1}$ that is obtained from~$F_t$.
We first assume that~$j=n$, and we distinguish whether~$j=n$ is a root in~$F_t$ or not (cases~(i) and~(ii) below).
We will see that the case $j<n$ is covered in case~(ii), by also considering the next step~$F_{t+1}\rightarrow F_{t+2}$.

{\bf Case~(i):}
If $j=n$ is not a root in~$F_t$, then by~(b) and the definition of step~R2 of Algorithm~R, the vertex~$j=n$ up-rotates again in the next step, i.e., $F_{t+1}$ is obtained from~$F_t$ by up-rotating~$j=n$, and the invariant~(b) remains satisfied.

{\bf Case~(ii):}
If $j=n$ is the root of an elimination tree in~$F_t$, then down-rotating~$n$ yields the already visited elimination forest~$F_{t-1}$.
Consequently, Algorithm~R considers vertices~$j<n$ for rotation.
By~(a), whenever a vertex~$j<n$ was rotated in the sequence~$\cF=(F_0,\ldots,F_t)$, the vertex~$n$ was a root or a leaf.
Consequently, deleting the vertex~$n$ from every elimination forest of this subsequence of~$\cF$, we obtain a sequence of elimination forests for~$G-n$ that Algorithm~R produces for the graph~$G-n$, and the rotation operation applied by the algorithm to reach the next elimination forest in this subsequence is the same as the one applied to reach~$F_{t+1}$, i.e., $p(F_{t+1})$ is unvisited in this subsequence if and only if $F_{t+1}$ is unvisited in the entire sequence (this is because $n$ is a root in~$F_t$).
Moreover, all elimination forests obtained by down-rotating the vertex~$n$ from the root until it is a leaf are unvisited, as they give the same elimination forest for~$G-n$ when deleting~$n$, i.e., in the step~$F_{t+1}\rightarrow F_{t+2}$ the algorithm down-rotates~$n$ and the invariant~(b) is maintained.

To complete the proof, it remains to consider the case where $F_t$ is obtained from~$F_{t-1}$ by down-rotating~$j\in\{s,\ldots,n\}$.
This argument is analogous to the previous one, by exchanging the words `up' and `down', as well as `root' and `leaf'.
\end{proof}

\begin{lemma}
\label{lem:clean-rots-j}
Let $G=([n],E)$ be a graph with simplicial vertices~$s,s+1,\ldots,n$.
Then for any initial elimination forest~$F_0\in\cE(G)$, whenever Algorithm~R performs an up- or down-rotation of a vertex~$j\in[n]$, then for every vertex~$k=\max\{j+1,s\},\ldots,n$ the vertices smaller than~$k$ in its elimination tree are either all descendants of~$k$, or none of them is a descendant of~$k$.
\end{lemma}

Lemma~\ref{lem:clean-rots-j} implies Lemma~\ref{lem:clean-rots} stated in Section~\ref{sec:zigzag-chordal}.

\begin{proof}
We say that a vertex~$j\in[n]$ in an elimination tree for~$G$ is \emph{up-blocked}, if $j$ is the root of the tree or the parent of~$j$ is larger than~$j$.
Similarly, we say that~$j$ is \emph{down-blocked}, if $j$ is a leaf of the elimination tree or all children of~$j$ are larger than~$j$.
As a consequence of Lemma~\ref{lem:project}~(b) and the choice of the largest vertex to rotate in step~R2 of Algorithm~R, in every step where the algorithm up- or down-rotates a vertex~$j\leq [n]$, all vertices~$k=\max\{j+1,s\},\ldots,n$ are up-blocked or down-blocked in the current elimination forest.

It follows inductively for $k=n,n-1,\ldots,\max\{j+1,s\}$ that if $k$ is up-blocked, then the vertices smaller than~$k$ in its elimination tree are all descendants of~$k$ (recall Lemma~\ref{lem:smaller-child-j}), whereas if $k$ is down-blocked, then none of them is a descendant of~$k$.
\end{proof}

The next lemma isolates the reason for failure of Algorithm~R.
It is motivated by the example of the 4-cycle shown in Figure~\ref{fig:C4}, where the algorithm terminates with the elimination tree~$T_{16}$ due to ambiguity about several possible rotations of the largest vertex~4 in this tree.

\begin{lemma}
\label{lem:non-peo-4}
Let $G=([n],E)$ be a connected non-PEO graph, and let $\ell$ be the largest non-simplicial vertex.
If $\ell\geq 4$, then Algorithm~R terminates without exhaustively generating~$\cE(G)$, for any initial elimination forest~$F_0\in\cE(G)$.
\end{lemma}

\begin{proof}
Note that $G^{[\ell]}$ is connected and has~$\ell\geq 4$ vertices.
As $\ell$ is not simplicial, it has two neighbors~$i,j<\ell$ that are not joined by an edge.
Moreover, as $\ell\geq 4$ and~$G^{[\ell]}$ is connected, there is another vertex~$k<\ell$ connected to a nonempty subset of~$\{i,j,\ell\}$; see Figure~\ref{fig:non-chordal}~(a).
We consider the elimination forest~$F$ of~$G^{[\ell]}$ obtained by removing the vertices~$[\ell]\setminus \{i,j,k,\ell\}$ in any order, and then $k,\ell,i,j$ in this order; see Figure~\ref{fig:non-chordal}~(b).
Moreover, we let $F^i$, $F^j$ and~$F^k$ denote the elimination forests obtained from~$F$ by rotating the tree edge~$\{\ell,i\}$, $\{\ell,j\}$ or~$\{\ell,k\}$, respectively.

\begin{figure}[b!]
\centering
\includegraphics{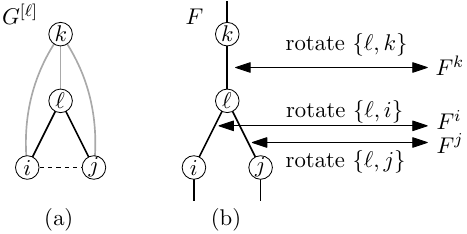}
\caption{The elimination forest~$F$ has three neighbors~$F^i,F^j,F^k$ obtained by up- or down-rotations of~$\ell$.
The dashed line indicates that $i$ and~$j$ are not adjacent in~$G$, and the three gray lines incident with~$k$ indicate that at least one of them is an edge in~$G$.}
\label{fig:non-chordal}
\end{figure}

Let $\cF=(F_0,\ldots,F_T)$ be the sequence of elimination forests for~$G$ produced by Algorithm~R initialized with~$F_0$ until the algorithm terminates.

By the choice of~$\ell$, all vertices $\ell+1,\ldots,n$ are simplicial.
Applying Lemma~\ref{lem:clean-rots-j} with $s:=\ell+1$ shows that in every step~$F_{t-1}\rightarrow F_t$ where Algorithm~R up- or down-rotates a vertex~$j\leq \ell$, then for every vertex~$k=\ell+1,\ldots,n$ the vertices smaller than~$k$ in its elimination tree are either all descendants of~$k$, or none of them is a descendant of~$k$, i.e., these rotations are clean (recall the definition of clean rotations given after Lemma~\ref{lem:clean-jumps}).
By considering the subsequence of~$\cF$ of those elimination forests and deleting the vertices~$n,n-1,\ldots,\ell+1$, we obtain a sequence~$\cF'$ of elimination forests for~$G^{[\ell]}$ that Algorithm~R produces for the graph~$G^{[\ell]}$, and the rotation operation applied by the algorithm to reach the next elimination forest in~$\cF'$ is the same as the one applied to reach the next elimination forest in the corresponding step in~$\cF$.

Suppose for the sake of contradiction that the sequence~$\cF$ contains every elimination forest of~$\cE(G)$ exactly once, then in particular~$\cF'$ must contain~$F$, $F^i$, $F^j$ and~$F^k$, in any order.
If $F$ comes before at least two elimination forests from the set~$\cR:=\{F^i,F^j,F^k\}$, then the algorithm would have terminated when encountering~$F$, as there are two possible edges in~$F$ with end vertex~$\ell$ such that rotating them leads to previously unvisited elimination forests from~$\cR$, a contradiction.
On the other hand, if $F$ comes after at least two elimination forests from~$\cR$, then when encountering the first of them, the algorithm would have rotated an edge with end vertex~$\ell$, and since the algorithm did not terminate because of ambiguity, it must have produced~$F$ next, a contradiction because~$F$ appears only after the second elimination forest from~$\cR$.

This completes the proof.
\end{proof}

The smallest connected non-PEO graph~$G=([n],E)$ not captured by Lemma~\ref{lem:non-peo-4} is the path on three vertices with the largest vertex~3 in the middle.
In fact, Algorithm~R succeeds in exhaustively generating all five elimination trees for this graph, despite the fact that this is \emph{not} a PEO; see Figure~\ref{fig:P3-non-chordal}.

\begin{figure}[h!]
\centering
\includegraphics[page=4]{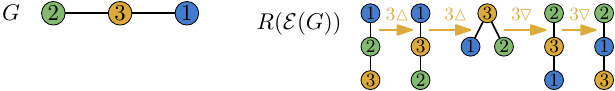}
\caption{A non-PEO graph~$G$ and the Gray code of its elimination trees obtained from Algorithm~R.
Note that in the third elimination tree, the vertex~3 has \emph{two} smaller children, namely~1 and~2, and that vertex~2 never rotates up or down (cf.\ Lemma~\ref{lem:smaller-child} and Figure~\ref{fig:P3-chordal}).
}
\label{fig:P3-non-chordal}
\end{figure}

\begin{theorem}
\label{thm:necessity}
If $G$ is not chordal, then for any ordering of the vertices of~$G$ and any initial elimination forest~$F_0\in\cE(G)$, Algorithm~R fails to exhaustively generate~$\cE(G)$.
\end{theorem}

\begin{proof}
Consider any ordering $1,\ldots,n$ of the vertices of~$G$.
As $G$ is not chordal, it has a minimal cycle of length at least four.
Let $\ell$ be the largest vertex in this cycle, which satisfies $\ell\geq 4$.
As the cycle is minimal, the two neighbors of~$\ell$ on the cycle are not adjacent, i.e., $\ell$ is not simplicial.
By Lemma~\ref{lem:non-peo-4}, Algorithm~R fails to exhaustively generate all elimination trees for the component containing~$\ell$.
\end{proof}

\section{Open problems}
\label{sec:open}

We conclude this paper with the following remarks and open problems.

\begin{itemize}[leftmargin=5mm, noitemsep, topsep=3pt plus 3pt]
\item
While we precisely characterized the class of graphs on which our method applies, one may wonder whether it could be applied to more general families of polytopes.
\emph{Nestohedra}, for instance, generalize graph associahedra, and can be defined as Minkowski sums of standard simplices corresponding to families of subsets of $\{1,2,\ldots,n\}$ known as building sets~\cite{MR2487491}.
The main property of building sets is that the union of two intersecting subsets must also be in the building set, which clearly holds for connected subsets of vertices of a graph.
In this special case, the building set is said to be graphical, and we recover the definition of graph associahedra given in Theorem~\ref{thm:Msum}.
Postnikov, Reiner, and Williams~\cite{MR2520477} define chordal nestohedra as a generalization of chordal graph associahedra, via the definition of chordal building sets.
We omit details here, but our generation algorithm should apply directly to chordal nestohedra, further extending its scope of applicability.
Considering arbitrary hypergraphs instead of building sets yields the class of \emph{hypergraphic polytopes}~\cite{aguiar_ardila_2023,MR3960512}, not all of which are Hamiltonian.
The question of the applicability of our generation algorithm in this wider setting is worth considering.

\item
Can we generalize the methods from this paper to efficiently compute Hamilton paths and cycles in graph associahedra~$\cA(G)$ for non-chordal graphs~$G$?
As a first step, one might try to tackle the case of $G$ being a cycle, i.e., cyclohedra.
As we saw for the 4-cycle in Figure~\ref{fig:C4}, and more generally for non-chordal graphs in the proof of Lemma~\ref{lem:non-peo-4}, the simple greedy rule of Algorithm~R will stop prematurely because of ambiguity.
To overcome this, a more global algorithmic control seems to be necessary, possibly using ingredients from Manneville and Pilaud's Hamiltonicity proof~\cite{MR3383157}.
However, the cycles resulting from their inductive gluing argument have a completely different structure from the ones produced by our Algorithm~R for chordal graphs, namely they are made of $n$ blocks and one of the $n$ vertices of the graph is root in all elimination trees of one block.
In contrast to that, the largest vertex~$n$ in our cycles constantly moves up and down the elimination tree, and is root for only two consecutive steps each.

\item
Another worthwile goal is to efficiently compute Hamilton cycles in graph associahedra~$\cA(G)$ for any chordal graph~$G$, specifically for graphs that do not satisfy the conditions of Theorems~\ref{thm:2conn} or~\ref{thm:peo}, in particular for trees~$G$.
Hurtado and Noy~\cite{MR1723053} provide a simple construction when $G$ is a path (i.e., the standard associahedron), which can probably be turned into an efficient algorithm and possibly be generalized.

\item
The function~$e(G)$ that counts the number of elimination forests for a graph~$G$, referred to as the $G$-Catalan number by Postnikov~\cite{MR2487491}, deserves further study; recall~\eqref{eq:eG}.
For example, what is the complexity of computing~$e(G)$?
This question is directly related to unranking and ranking in the orderings $R(\cE(G))$ computed by our algorithm, or to use our generation tree approach for random sampling.

\item
By the results of Barnard and McConville~\cite{MR4176852} and Lemma~\ref{lem:filled}, the equivalence relation~$\equiv_G$ is a lattice congruence of the weak order if and only if $G$ is a unit interval graph.
Their paper also raises the question: Under what conditions on~$G$ is~$\equiv_G$ a lattice?
\end{itemize}

\section{Acknowledgements}

We thank Vincent Pilaud for pointing us to the proof of Lemma~\ref{lem:even}, and we thank Petr Gregor for helpful comments on an earlier version of this paper.
We also thank the reviewers of this paper for their remarks and suggestions.

\bibliographystyle{alpha}
\bibliography{refs}

\end{document}